\documentclass[11pt,a4paper]{article}

\usepackage[T1]{fontenc}
\usepackage[utf8]{inputenc}
\usepackage[english]{babel}
\usepackage[a4paper,total={6in,8in}]{geometry}
\usepackage{amsmath,amssymb,amsfonts,amsthm,bm,mathtools}
\usepackage{graphicx}
\usepackage{xcolor}
\usepackage{authblk}
\usepackage[round,sort&compress]{natbib}
\usepackage{url}
\usepackage[colorlinks=true,linkcolor=blue,citecolor=blue,urlcolor=blue]{hyperref}
\graphicspath{{./}{Fig/}}
\theoremstyle{plain}
\newtheorem{theorem}{Theorem}
\newtheorem{proposition}[theorem]{Proposition}
\newtheorem{corollary}[theorem]{Corollary}
\newtheorem{lemma}{Lemma}
\theoremstyle{definition}
\newtheorem{assumption}{Assumption}
\newtheorem{definition}{Definition}

\theoremstyle{remark}
\newtheorem{remark}{Remark}

\newcommand{\taglio}[1]{}

\title{Standardized Sign-Flipping Score Tests for Fixed Effects in Multivariate Clustered Data}
\author[1]{\href{https://orcid.org/0000-0002-1141-3041}{Angela Andreella}\thanks{Corresponding author: \href{mailto:angela.andreella@unive.it}{angela.andreella@unive.it}.}}
\author[2,3]{\href{https://orcid.org/0000-0003-3181-8078}{Livio Finos}}
\affil[1]{Department of Economics, Ca' Foscari University of Venice, Cannaregio 873, 30100 Venezia, Italy}
\affil[2]{Department of Statistical Sciences, University of Padova, via Cesare Battisti 241, 35129 Padua, Italy}
\affil[3]{Padova Neuroscience Center, University of Padova, Via Orus 2/B, 35129 Padua, Italy}

\date{}
\hypersetup{
  pdftitle={Standardized Sign-Flipping Score Tests for Fixed Effects in Multivariate Clustered Data},
  pdfauthor={Angela Andreella and Livio Finos}
}

\begin{document}
\maketitle

\begin{abstract}
This study introduces \emph{clip}, a clusterwise conditional resampling procedure that utilizes effective scores and transformation-specific standardization. While a working covariance can enhance efficiency, correct specification is not required for asymptotic level control. Synchronized clusterwise transformations allow joint inference across responses without explicitly modeling the between-response covariance. For scalar tests under Gaussian errors, a fixed, correctly specified working covariance and suitable cluster designs, standardization yields a size error of $O(N^{-3/2})$, compared with a positive error of order $N^{-1}$ for the corresponding nonstandardized wild cluster bootstrap. Simulation studies demonstrate reliable type I error control and competitive power. An electroencephalography dataset is used to illustrate simultaneous inference across correlated outcomes.
\end{abstract}


\section{Introduction}
Linear mixed models (LMMs) are widely used to analyze clustered and repeated-measures data \citep{pinheiro2006mixed}. Their inference on population-level fixed effects can, however, be sensitive to the specification of the within-cluster covariance. Omitting relevant random slopes or variance components can distort model-based standard errors and tests, whereas fitting complex random-effects structures can be unstable in moderate samples \citep{barr2013random,bates2015parsimonious,matuschek2017balancing}. This issue should be distinguished from misspecifying only the shape of the random-effects distribution, which substantially does not affect fixed-effect estimation when the marginal mean and covariance are correctly represented \citep{mcculloch2011misspecifying,verbeke1997effect}. We focus on inference under a correctly specified marginal mean while allowing the within-cluster covariance to be misspecified.

The problem becomes more challenging with multivariate clustered responses, where simultaneous inference must accommodate both within-cluster and between-response dependence. A joint mixed model requires many covariance parameters, making estimation difficult unless some constraints are imposed. Fitting separate models for each response simplifies estimation but does not provide the joint null distribution for global tests or for simultaneous inference that exploits association among responses. Procedures such as the Holm correction \citep{holm1979simple} remain valid under arbitrary dependence when the marginal tests are valid, but can lose power since they do not use this association \citep{goeman2014multiple}.

Marginal and resampling-based methods offer alternatives to fully parametric mixed-model inference. Cluster-robust sandwich \citep{cameron2015practitioner} methods are asymptotically valid under within-cluster covariance misspecification, but can be anticonservative with few clusters, unequal cluster sizes, or influential clusters \citep{hansen2019asymptotic, mackinnon2023cluster}. Cluster-level resampling \citep{mackinnon2023cluster} can improve finite-sample type I error control, but its performance depends on the statistic's construction and studentization. For both approaches, the multivariate setting introduces the additional challenge of accounting for cross-response dependence while preserving these validity guarantees.

We then propose \emph{clip} (clusterwise sign-flipping), a conditional resampling test for fixed effects in multivariate clustered linear models. For each response, the score is orthogonalized with respect to nuisance scores and decomposed into clusterwise effective-score contributions, with transformation-specific standardization to account for variance changes induced by sign flipping. Independent random signs are applied across clusters, with the same sign used for all responses within each cluster to preserve cross-response dependence. The resulting standardized score vectors can be combined to test intersection hypotheses or used in multiple-testing procedures.

Our contributions are threefold. (1) We establish joint asymptotic validity for independent, heterogeneous clusters under a correctly specified marginal mean and suitable regularity conditions. A working covariance can improve efficiency without needing to be correctly specified for validity. (2) Using a Gaussian approximation, we bound the difference between the rejection probability and the attainable significance level, accounting for unequal cluster contributions. (3) Relating the effective-score test \citep{hemerik2020robust} to the restricted wild (cluster) bootstrap \citep{wu1986jackknife, kline2012score}, we show that transformation-specific standardization removes the leading size distortion of the corresponding nonstandardized procedures. For scalar tests with Gaussian errors, a fixed, correctly specified working covariance, and suitable cluster designs, the size error decreases from $O(N^{-1})$ to $O(N^{-3/2})$, with a fixed number of transformations.

The method is implemented in the \texttt{R} package \texttt{remmm} (\url{https://github.com/angeella/remmm}).

\section{Related work}\label{sec:related_works}

Multivariate clustered outcomes require accounting simultaneously for within-cluster and between-response dependence. Parametric approaches include conditional models, latent-variable formulations, copula models, pairwise likelihood, and multivariate mixed models; see \citet{fieuws2004joint} and \citet{verbeke2014analysis}. These approaches are attractive when the joint dependence model is itself of scientific interest, but they retain sensitivity to how that model is specified and estimated.

A different strategy specifies the marginal mean and treats covariance as a nuisance. Generalized estimating equations (GEE) with cluster-robust sandwich inference are the canonical example \citep{liang1986longitudinal}, with multiple-marginal model approaches fitting responses separately and estimating the joint covariance \citep{jensen2018comparison,pallmann2018simultaneous}. These methods avoid a full multivariate likelihood but require estimating a $M\times M$ covariance/correlation matrix, which can become unstable when the number of responses $M$ is large relative to the number of independent clusters. Controlling type I error can also fail with few, unequal, or high-leverage clusters \citep{cameron2015practitioner,hansen2007asymptotic,hansen2019asymptotic}. Small-sample and bias-reduction corrections can improve performance \citep{bell2002bias,tipton2015small,pustejovsky2018small,mackinnon2023cluster}, although a single correction is not uniformly preferable.

Resampling-based methods \citep{pesarin2001multivariate} provide a robust alternative approach to inference. Wild residual resampling \citep{wu1986jackknife} extends to clustered data through common bootstrap weights within each cluster \citep{kline2012score}. The wild (cluster) bootstrap has several versions, with different residual adjustments and variance estimators \citep{mackinnon2023cluster,mackinnon2023fast}. Section~\ref{sec:ks_equivalence} shows that the effective score test \citep{hemerik2020robust} and the null-restricted scalar LM score bootstrap of \citet{wu1986jackknife} are permutationally equivalent \citep{pesarin2001multivariate}, i.e., they produce the same $p$-values, under the same Rademacher transformation. This is also true for clustered data \citep{kline2012score}. For scalar tests, \emph{clip} introduces transformation-specific standardization \citep{de2022inference} into this construction. Section~\ref{sec:gaussian-rate} shows how this standardization removes the leading size distortion of the corresponding nonstandardized procedures under suitable Gaussian-model and cluster-design conditions.

Related randomization methods have also been proposed specifically for mixed models. \citet{kherad2010exact} and \citet{lee2012permutation} consider tests based on transformed or permuted residuals. These methods are useful in particular designs, but require either restrictive exchangeability conditions or repeated model fitting, and do not provide a general multivariate construction of the form considered here. \citet{finos2014permutation} propose an approximate permutation method for multivariate generalized linear mixed models; its calibration involves estimated covariance components linking the covariates and random effects.

Methodologically, \emph{clip} builds on effective-score sign transformations for independent observations \citep{hemerik2020robust}, their transformation-specific standardization under variance misspecification \citep{de2022inference}, and the multivariate extension based on synchronized transformations across outcomes \citep{de2024permutation}. We extend these ideas to independent, possibly heterogeneous clusters with arbitrary within-cluster dependence. We establish the validity of the resulting clusterwise multivariate construction under unequal and potentially increasing cluster sizes, nuisance-parameter estimation, and estimated working weights, and derive a quantitative Gaussian-approximation bound.

\section{Framework and assumptions}
\label{sec:framework}

For cluster $j=1,\ldots,N$, let $\boldsymbol{Y}_j=(y_{ijm})\in\mathbb{R}^{n_j\times M}$ denote the cluster-level response matrix, whose $i$th row contains the $M$ outcomes observed at occasion $i=1,\ldots,n_j$. For outcome $m=1,\ldots,M$, let $\boldsymbol y_{jm}\in\mathbb R^{n_j}$ denote the $m$th column of $\boldsymbol Y_j$, and consider the working mixed-effects representation \begin{equation} \label{eq:model} \boldsymbol{y}_{jm} = \boldsymbol{x}_{jm}\beta_m + Z_{jm}\boldsymbol{\gamma}_m + Q_{jm}\boldsymbol{b}_{jm} + \boldsymbol{\varepsilon}_{jm}. \end{equation} Here, $\boldsymbol{x}_{jm}\in\mathbb{R}^{n_j}$ is the covariate of interest, $Z_{jm}\in\mathbb{R}^{n_j\times q_m}$ contains nuisance covariates, and $Q_{jm}\in\mathbb{R}^{n_j\times r_m}$ is a generic random-effects design matrix. The latent vector $\boldsymbol{b}_{jm}\in\mathbb{R}^{r_m}$ can contain random intercepts, random slopes, or other cluster-specific components, and $\boldsymbol{\varepsilon}_{jm}\in\mathbb{R}^{n_j}$ denotes the remaining within-cluster error. Both are taken to have mean zero and finite second moments. Within a cluster, components corresponding to different outcomes may be dependent or share common latent factors. Throughout, we treat the design matrices as fixed, and all expectations, variances, and covariances are understood conditionally on the observed design. Representation~\eqref{eq:model} is used only to motivate working covariance choices. The inferential results do not require this decomposition, its distributional assumptions, or its covariance structure to be correctly specified; see Remark \ref{rem:working_weights}.

For a fixed $\emptyset \ne L\subseteq\{1,\ldots,M\}$, we consider the intersection null hypothesis \begin{equation} \label{eq:null} H_0=\bigcap_{l\in L}H_{0l}, \qquad H_{0l}:\beta_l=\beta_l^0. \end{equation} The choice $L=\{1,\ldots,M\}$ gives a global multivariate test, while $|L|=1$ gives the univariate case. Throughout, $|L|$ is fixed as $N\to\infty$.

Let $V_{jl}=\operatorname{Var}(\boldsymbol{y}_{jl})$ denote the true marginal within-cluster covariance. We introduce a symmetric positive-definite working precision matrix $W_{jl}$. When $V_{jl}$ is nonsingular, $W_{jl}$  can be chosen to approximate $V_{jl}^{-1}$, but this is not required for validity. In this section, $W_l$ is fixed or measurable with respect to the observed design; we consider estimated working weights in Section~\ref{sec:clustertest}.

For fixed $W_l$, define the generalized least-squares criterion $\mathcal Q_l(\beta_l,\boldsymbol{\gamma}_l) =\frac{1}{2}\sum_{j=1}^N (\boldsymbol{y}_{jl}-\boldsymbol{x}_{jl}\beta_l-Z_{jl}\boldsymbol{\gamma}_l)^\top W_{jl} (\boldsymbol{y}_{jl}-\boldsymbol{x}_{jl}\beta_l-Z_{jl}\boldsymbol{\gamma}_l)$. Under $H_{0l}$ \eqref{eq:null}$$\widehat{\boldsymbol{\gamma}}_l =(\sum_{j=1}^N Z_{jl}^\top W_{jl}Z_{jl})^{-1} \sum_{j=1}^N Z_{jl}^\top W_{jl} (\boldsymbol{y}_{jl}-\boldsymbol{x}_{jl}\beta_l^0),$$ provided that $\sum_{j=1}^N Z_{jl}^\top W_{jl}Z_{jl}$ is nonsingular. Let $n=\sum_{j=1}^N n_j$, $\boldsymbol{y}_l=(\boldsymbol{y}_{1l}^\top,\ldots, \boldsymbol{y}_{Nl}^\top)^\top$, $\boldsymbol{x}_l=(\boldsymbol{x}_{1l}^\top,\ldots, \boldsymbol{x}_{Nl}^\top)^\top$, and $Z_l=(Z_{1l}^\top,\ldots,Z_{Nl}^\top)^\top$. Let $W_l$ be block diagonal with diagonal blocks $W_{1l},\ldots,W_{Nl}$, and define $H_l=W_l^{1/2}Z_l(Z_l^\top \allowbreak  W_lZ_l)^{-1} Z_l^\top W_l^{1/2}$. The profiled score for $\beta_l$ under $H_{0l}$ is \begin{equation} \label{eq:score} \widehat U_l = \boldsymbol{x}_l^\top W_l (\boldsymbol{y}_l-\boldsymbol{x}_l\beta_l^0 -Z_l\widehat{\boldsymbol{\gamma}}_l) = \boldsymbol{x}_l^\top W_l^{1/2} (I-H_l)W_l^{1/2} (\boldsymbol{y}_l-\boldsymbol{x}_l\beta_l^0). \end{equation} Equation~\eqref{eq:score} is the starting point for the clusterwise score construction used below.

\begin{assumption} \label{ass1} The cluster-level response matrices $\boldsymbol{Y}_1,\ldots,\boldsymbol{Y}_N$ are independent. Moreover, for every $l\in L$ and $j=1,\ldots,N$, $\operatorname{E}(\boldsymbol{y}_{jl}) =\boldsymbol{x}_{jl}\beta_l+Z_{jl}\boldsymbol{\gamma}_l$. \end{assumption}

Assumption~\ref{ass1} requires correct specification of the fixed-effects mean but imposes no parametric distribution or covariance model for the latent within-cluster dependence.  See also Remark \ref{rem:working_weights} and the comments above.

Let $\boldsymbol{\gamma}_l^\ast$ denote the true nuisance parameter under $H_{0l}$ and define $\boldsymbol{\mu}_l^\ast = \boldsymbol{x}_l\beta_l^0 + Z_l\boldsymbol{\gamma}_l^\ast$. The corresponding oracle effective score is $U_l^\ast = \boldsymbol{x}_l^\top W_l^{1/2}(I-H_l)W_l^{1/2} (\boldsymbol{y}_l-\boldsymbol{\mu}_l^\ast)$. Since $(I-H_l)W_l^{1/2}Z_l=0$, $U_l^\ast = \boldsymbol{x}_l^\top W_l^{1/2}(I-H_l)W_l^{1/2} (\boldsymbol{y}_l-\boldsymbol{x}_l\beta_l^0) = \widehat U_l$. Thus, nuisance-parameter estimation does not affect the aggregate observed score. It does, however, affect its individual clusterwise contributions, as shown below.

\begin{definition} \label{def:cluster_score} For $l\in L$, let $\widehat{\boldsymbol{\mu}}_l =\boldsymbol{x}_l\beta_l^0+Z_l\widehat{\boldsymbol{\gamma}}_l$. The clusterwise contribution to the effective score in \eqref{eq:score} is \begin{equation} \label{eq:cluster_score} \widehat\zeta_{jl} = \left[(I-H_l)W_l^{1/2}\boldsymbol{x}_l\right]_j^\top \left[W_l^{1/2} (\boldsymbol{y}_l-\widehat{\boldsymbol{\mu}}_l)\right]_j, \end{equation} where $[\cdot]_j$ denotes the subvector corresponding to cluster $j$. The oracle contribution $\zeta_{jl}^\ast$ is defined as in \eqref{eq:cluster_score}, replacing $\widehat{\boldsymbol{\mu}}_l$ by $\boldsymbol{\mu}_l^\ast$. \end{definition}

By construction, $\widehat U_l=\sum_{j=1}^N\widehat\zeta_{jl}$ and $U_l^\ast=\sum_{j=1}^N\zeta_{jl}^\ast$. Under Assumption~\ref{ass1} and $H_{0l}$, $\operatorname{E}(\zeta_{jl}^\ast)=0$ for every cluster $j$ and for any fixed or design-measurable positive-definite working weight. Thus, null centering of the oracle score does not depend on correct specification of the within-cluster covariance. The working weights affect efficiency and power, while the clusterwise sign transformations introduced in
Section~\ref{sec:clustertest} yield asymptotically valid tests under the conditions stated there.

We now state regularity conditions on the oracle clusterwise contributions. Define $\boldsymbol{\zeta}_j^\ast=(\zeta_{jl}^\ast:l\in L)^\top$, $v_{Nl}^2=\sum_{j=1}^N\operatorname{Var}(\zeta_{jl}^\ast)$,
and $\Lambda_N=\operatorname{diag}(v_{Nl}^2:l\in L)$,
assuming $v_{Nl}^2>0$ for all sufficiently large $N$.

\begin{assumption} \label{ass:lyapunov} There exists $d>2$ such that \begin{equation} \label{eq:lyapunov} \sum_{j=1}^N \operatorname{E} \left[ \left\| \Lambda_N^{-1/2}\boldsymbol{\zeta}_j^\ast \right\|_2^d \right] \longrightarrow 0. \end{equation} \end{assumption}

Condition~\eqref{eq:lyapunov} is a multivariate Lyapunov condition. It allows heterogeneous and increasing cluster sizes while ruling out asymptotically dominant cluster contributions. Conditions of this type are standard for heterogeneous clustered arrays \citep{djogbenou2019asymptotic,hansen2019asymptotic}. For example, suppose that, uniformly in $j$ and $l\in L$, $\operatorname{E}|\zeta_{jl}^\ast|^d\lesssim n_j^{d/2}$ and $v_{Nl}^2\asymp\sum_{j=1}^N n_j=n$. Then the left-hand side of \eqref{eq:lyapunov} is of order at most $(\max_j n_j/n)^{d/2-1}$, so $\max_j n_j/n\to0$ is sufficient for Assumption~\ref{ass:lyapunov}. Under stronger persistent within-cluster dependence, suppose instead that $\operatorname{E}|\zeta_{jl}^\ast|^d\lesssim n_j^d$ and $v_{Nl}^2\asymp\sum_j n_j^2$. In this case, $\max_j n_j^2/\sum_{h=1}^N n_h^2\to0$ is sufficient. Thus, stronger within-cluster dependence places a more stringent restriction on cluster-size imbalance, although individual cluster sizes may still increase with $N$. 
Let $\Omega_N=\sum_{j=1}^N \operatorname{Cov}(\boldsymbol{\zeta}_j^\ast)$ and define $R_N=\Lambda_N^{-1/2}\Omega_N\Lambda_N^{-1/2}$.

\begin{assumption} \label{ass:correlation} For the fixed index set $L$, there exists a symmetric positive semidefinite matrix $R\in\mathbb{R}^{|L|\times|L|}$ with unit diagonal such that, as $N\to\infty$, $R_N\to R$. \end{assumption}

Assumption~\ref{ass:correlation} imposes no parametric structure on the within-cluster and cross-outcome dependence. It requires that the correlation matrix of the aggregated oracle scores stabilizes as the number of clusters increases.

Under $H_0$ and Assumptions~\ref{ass1}--\ref{ass:correlation}, as $N\to\infty$, the multivariate Lyapunov central limit theorem yields $\Lambda_N^{-1/2}\sum_{j=1}^N\boldsymbol{\zeta}_j^\ast \xrightarrow{d}N_{|L|}(0,R)$, i.e., the oracle reference distribution for the \emph{clip} test introduced in the next section. We place all proofs in the supplementary materials.

\section{Multivariate clusterwise score test}
\label{sec:clustertest}

We first define the marginal score statistics used to test \eqref{eq:null}. For each $l\in L$, let
$A_l=(I-H_l)W_l^{1/2}\boldsymbol{x}_l$,
$a_l=\|A_l\|_2$, and
$\widehat{\sigma}_l=
(\sum_{j=1}^N\widehat{\zeta}_{jl}^2)^{1/2}/a_l$.

\begin{definition} \label{def:marginal_score} For each $l\in L$, the standardized effective-score statistic is \begin{equation} \label{eq:score_standardized} \widehat S_l = \frac{\sum_{j=1}^N\widehat{\zeta}_{jl}} {a_l\widehat{\sigma}_l}. \end{equation} \end{definition}

By construction, equation~\eqref{eq:score_standardized} can equivalently be written as $\widehat S_l= \sum_{j=1}^N\widehat{\zeta}_{jl}/ (\sum_{j=1}^N\widehat{\zeta}_{jl}^2)^{1/2}$. The factor $\widehat{\sigma}_l$ provides an outcome-specific standardization that places the marginal statistics on comparable scales across outcomes.

Having defined the marginal statistics, we combine them to test \eqref{eq:null}.

\begin{definition} \label{def:multivariate_statistic} For a two-sided alternative, let $\psi:\mathbb{R}_+^{|L|}\to\mathbb{R}$ be continuous and non-decreasing in each argument. The multivariate test statistic is \begin{equation} \label{eq:combined_statistic} \widehat T = \psi\bigl(|\widehat S_l|:l\in L\bigr). \end{equation} \end{definition}

Different choices of $\psi$ are possible \citep{pesarin2001multivariate, vesely2023permutation}. Maximum-type functions are suitable when the alternative is concentrated in one or a few outcomes, whereas sum- or quadratic-type functions are more sensitive when moderate effects are distributed across several outcomes. In our implementation, we use the max-$T$ combining function \citep{westfall1993resampling}, i.e., $\psi(z_l:l\in L)=\max_{l\in L}z_l$.

To approximate the null distribution of \eqref{eq:combined_statistic}, we apply clusterwise sign transformations to the effective-score contributions. Let $\mathcal F$ denote the collection of block-diagonal matrices $F=\operatorname{diag}(s_1I_{n_1},\ldots,s_NI_{n_N})$, with $s_j\in\{-1,+1\}$. We draw $B$ transformations $F_1,\ldots,F_B$, setting $F_1=I$ and sampling $F_2,\ldots,F_B$ independently and uniformly from $\mathcal F$, independently of the data. Throughout, $B\geq2$ and $0<\alpha<1$. For transformation $F_b$, let $s_{bj}\in\{-1,1\}$ denote the sign assigned to cluster $j$ by $F_b$ and define $a_l(F_b)= \|(I-H_l)F_b(I-H_l)W_l^{1/2}\boldsymbol{x}_l\|_2$ with $a_l>0$. The transformed standardized score is \begin{equation} \label{eq:transformed_score} \widehat S_l(F_b) = \frac{\sum_{j=1}^N s_{bj}\widehat{\zeta}_{jl}} {a_l(F_b)\widehat{\sigma}_l}, \qquad l\in L. \end{equation} For $F_1=I$, Equation~\eqref{eq:transformed_score} reduces to $\widehat S_l(F_1)=\widehat S_l$. The factor $a_l(F_b)$ provides a transformation-specific standardization accounting for the scale change induced by nuisance-score projection, whereas $\widehat{\sigma}_l$ is common to all transformations within outcome $l$. The same clusterwise signs $s_{bj}$ are used synchronously across all outcomes, preserving cross-outcome dependence within each transformed replicate. For each transformation, define $\widehat T(F_b)= \psi(|\widehat S_l(F_b)|:l\in L)$, so that $\widehat T(F_1)=\widehat T$. The $p$-value is $p=B^{-1}\sum_{b=1}^B \mathbf{1}\{\widehat T(F_b)\geq\widehat T(F_1)\}$, and $H_0$ is rejected at level $\alpha$ whenever $p\leq\alpha$.

We next give conditions under which replacing the oracle clusterwise contributions by their estimated counterparts is asymptotically negligible.

\begin{assumption} \label{ass:nuisance_rate} For every $l\in L$, the nuisance dimension $q_l$ is fixed and there exists a deterministic sequence $r_{Nl}\to\infty$ such that $r_{Nl}\|\widehat{\boldsymbol{\gamma}}_l- \boldsymbol{\gamma}_l^\ast\|=O_p(1)$. \end{assumption}

The rate $r_{Nl}$ depends on the strength of within-cluster dependence, the cluster-size distribution, and the estimation method. Rates of order $\sqrt{N}$ arise under persistent within-cluster dependence, whereas faster rates can be available under weaker dependence \citep{djogbenou2019asymptotic,hansen2007asymptotic,hansen2019asymptotic}.

For $j=1,\ldots,N$ and $l\in L$, define $c_{jl}= Z_{jl}^\top W_{jl}^{1/2} [A_l]_j \in\mathbb{R}^{q_l}$. By Definition~\ref{def:cluster_score}, $\widehat{\zeta}_{jl}-\zeta_{jl}^\ast =-c_{jl}^\top (\widehat{\boldsymbol{\gamma}}_l-\boldsymbol{\gamma}_l^\ast)$. Thus, $c_{jl}$ measures the sensitivity of the $j$th clusterwise score contribution to nuisance-parameter estimation.

\begin{assumption} \label{ass:nuisance_sensitivity} For every $l\in L$, $\{\sum_{j=1}^N\|c_{jl}\|_2^2\}^{1/2}/ (r_{Nl}v_{Nl})\to0$. \end{assumption}

For $F_1$, nuisance estimation cancels exactly because $Z_l^\top W_l^{1/2}A_l=0$. For $b\geq2$, taking expectation with respect to the random signs gives $\operatorname{E}_F\| \sum_{j=1}^N s_{bj}c_{jl}\|_2^2 =\sum_{j=1}^N\|c_{jl}\|_2^2$. Hence, since $B$ and $|L|$ are fixed, Assumptions~\ref{ass:nuisance_rate}--\ref{ass:nuisance_sensitivity} imply $\max_{1\leq b\leq B}\max_{l\in L} v_{Nl}^{-1} |\sum_{j=1}^N s_{bj} (\widehat{\zeta}_{jl}-\zeta_{jl}^\ast)| =o_p(1)$. When cluster sizes are uniformly bounded, and the design and working weights are uniformly controlled, we have $(\sum_j\|c_{jl}\|_2^2)^{1/2}=O(\sqrt N)$ and $v_{Nl}\asymp\sqrt N$, so Assumption~\ref{ass:nuisance_sensitivity} follows from $r_{Nl}\to\infty$. More generally, increasing and heterogeneous cluster sizes are allowed whenever the sensitivity of the transformed score grows sufficiently slowly relative to $r_{Nl}v_{Nl}$.

The final condition concerns the standardization in \eqref{eq:transformed_score}.

\begin{assumption} \label{ass:standardization} For every $l\in L$, with $B$ fixed, $\max_{1\leq b\leq B} |v_{Nl}/\{a_l(F_b)\widehat{\sigma}_l\}-1| \xrightarrow{p}0$. \end{assumption}

Assumption~\ref{ass:standardization} requires the complete transformation-specific denominator in \eqref{eq:transformed_score} to be uniformly equivalent to the oracle standard deviation $v_{Nl}$ over the fixed number of transformations. %It does not require the working covariance to be correctly specified.

\begin{lemma} \label{lem:exchangeability} Let $B$ and $|L|$ be fixed, and let $M_N$ be the $B\times|L|$ matrix with entries $[M_N]_{bl}=\widehat S_l(F_b)$. Under $H_0$ and Assumptions~\ref{ass1}--\ref{ass:standardization}, as $N\to\infty$, $M_N\xrightarrow{d}M$, where the rows of $M$ are independent and identically distributed $N_{|L|}(0,R)$ random vectors. \end{lemma}

Lemma~\ref{lem:exchangeability} establishes the key property of the construction: the observed and transformed score vectors are asymptotically row-exchangeable. Synchronous signs preserve the cross-outcome covariance within each row, while independent sign vectors across transformations yield asymptotic independence between rows.

\begin{theorem} \label{thm:clip_validity} Let $\widehat T(F_b)$ be defined by \eqref{eq:combined_statistic}. Under the conditions of Lemma~\ref{lem:exchangeability}, suppose that $\psi(|Z_l|:l\in L)$ has a continuous distribution for $Z\sim N_{|L|}(0,R)$. Then, under $H_0$, the test that rejects when $p\leq\alpha$ satisfies $\Pr(p\leq\alpha)\to \lfloor\alpha B\rfloor/B\leq\alpha$ as $N\to\infty$. \end{theorem}

Theorem~\ref{thm:clip_validity} provides an asymptotically valid local test for every $\emptyset \ne L\subseteq\{1,\ldots, M\}$. In particular, taking $L=\{1,\ldots,M\}$ yields asymptotically weak familywise error rate (FWER) control under the complete null hypothesis. More generally, once a valid multivariate score test and its joint resampling distribution are available, we can refer to closed-testing procedures with asymptotically strong FWER control, including methods that provide confidence bounds for the true discovery proportion \citep{andreella2023permutation,vesely2023permutation,blain2022notip}. This is relevant because closed testing procedures are, in a precise sense, optimal: any multiple-testing procedure can be rewritten in terms of closed testing and can be uniformly improved by its closed-testing version \citep{goeman2021only}.

So far, the working weights have been treated as fixed or measurable with respect to the design. In practice, $W_l$ will often depend on unknown covariance parameters and will be replaced by an estimator $\widehat W_l$. The working weights are estimated once from the original data and then held fixed across all transformations.

To make this distinction explicit, let $\widehat S_l(F_b;W_l)$ denote the statistic in \eqref{eq:transformed_score} computed using the target working weight $W_l$, and let $\widehat S_l(F_b;\widehat W_l)$ denote the statistic obtained by replacing $W_l$ by $\widehat W_l$ in all weight-dependent quantities.
\begin{assumption} \label{ass:estimated_weights} For fixed $B$ and $|L|$, $\max_{1\leq b\leq B}\max_{l\in L} |\widehat S_l(F_b;\widehat W_l) -\widehat S_l(F_b;W_l)| \xrightarrow{p}0$. \end{assumption}

Assumption~\ref{ass:estimated_weights} requires estimation of the working weights to produce only an asymptotically negligible perturbation of the entire transformed score matrix. Importantly, $\widehat W_l$ need not consistently estimate the inverse true covariance; it need only converge sufficiently well to its working target.

\begin{corollary} \label{cor:estimated_weights} Under $H_0$, Assumptions~\ref{ass1}--\ref{ass:estimated_weights}, and the continuity condition in Theorem~\ref{thm:clip_validity}, the test based on $\widehat S_l(F_b;\widehat W_l)$ has rejection probability converging to $\lfloor\alpha B\rfloor/B\leq\alpha$. \end{corollary}

Corollary~\ref{cor:estimated_weights} applies to working covariances estimated either under the null or from an unrestricted model. To see the latter, let $\widehat W_l^{(0)}$ and $\widehat W_l^{(u)}$ denote the corresponding null-based and unrestricted estimates. If Assumption~\ref{ass:estimated_weights} holds for $\widehat W_l^{(0)}$ and, under $H_0$, $\max_{b,l} |\widehat S_l(F_b;\widehat W_l^{(u)}) -\widehat S_l(F_b;\widehat W_l^{(0)})| =o_p(1)$, then the same conclusion follows for $\widehat W_l^{(u)}$. A sufficient setting is one in which the null-based and unrestricted covariance estimators converge under $H_0$ to the same deterministic working target and the transformed score map is locally stable with respect to the working weights. Standard fixed-dimensional ML or REML theory may establish the first requirement under appropriate model-specific regularity conditions \citep{jiang1996reml,richardson1994asymptotic}; the required stability of the transformed statistics is separately encoded in Assumption~\ref{ass:estimated_weights}. Estimating the working covariance under the unrestricted model may be advantageous under alternatives because a covariance estimate obtained under the null can absorb part of the fixed-effect signal. 

\begin{remark}\label{rem:working_weights}
Misspecification of the working covariance may affect efficiency but preserves asymptotic level control. When reliable covariance information is unavailable, working independence, $W_{jl}=I_{n_j}$, provides a simple default, as in GEE \citep{liang1986longitudinal}. Alternatives include diagonal weights accounting for marginal heteroscedasticity and weights derived from a mixed-effects model, possibly estimated without imposing the null. These choices may improve efficiency and power, but no uniform power gain is guaranteed. Section~\ref{sec:simulations} examines the finite-sample implications of these choices for calibration and power.
\end{remark}

\subsection{Connections with the wild cluster bootstrap}
\label{sec:ks_equivalence}

We compare procedures with the same hypothesis, nuisance design,
null-restricted residuals, cluster partition, working weights,
and Rademacher signs. The working weights are held fixed across
transformations, including when estimated from the data.
For each outcome $l$, write the transformed score numerator as
$$
\widehat P_l(F_b)
=A_l^\top F_bW_l^{1/2}
(\boldsymbol{y}_l-\widehat{\boldsymbol{\mu}}_l)
=\sum_{j=1}^N s_{bj}\widehat\zeta_{jl}.
$$
The following proposition establishes the permutational equivalence
\citep{pesarin2001multivariate} between the clusterwise
effective-score test \citep{hemerik2020robust} and a restricted
wild (cluster) bootstrap based on the regression coefficient
\citep{wu1986jackknife}.

\begin{proposition}
\label{prop:matched-wcb}
For each outcome $l$, define the bootstrap response
$\boldsymbol{y}_l^\ast(F)
=\widehat{\boldsymbol{\mu}}_l
+F(\boldsymbol{y}_l-\widehat{\boldsymbol{\mu}}_l)$,
and let $\widehat{\beta}_l^\ast(F)$ be the coefficient of
$\boldsymbol{x}_l$ obtained by GLS regression of
$\boldsymbol{y}_l^\ast(F)$ on $(\boldsymbol{x}_l,Z_l)$,
using the fixed working weights $W_l$. Then
$\widehat{\beta}_l^\ast(F)-\beta_l^0
=\frac{\widehat P_l(F)}{a_l^2}$,
with $\widehat{\beta}_l^\ast(I)=\widehat{\beta}_l$,
the GLS estimate from the observed data.
Consequently, the restricted wild (cluster) bootstrap
based on $|\widehat{\beta}_l^\ast(F)-\beta_l^0|$
and the clusterwise effective-score test give identical
two-sided marginal $p$-values when they use the same
transformations and tie convention to test $H_0:\ \beta_l=\beta_l^0$.
\end{proposition}

A further equivalence holds with the null-restricted scalar
LM score bootstrap of \citet{kline2012score}.
The effective-score statistic is
$E_l(F_b)=\frac{\widehat P_l(F_b)}{a_l\widehat\sigma_l}$. With the same projected score contributions, the LM score
bootstrap statistic is
\begin{equation}
\label{eq:ks_monotone}
K_l(F_b)
=\frac{\widehat P_l(F_b)^2}
{\frac{N}{N-1}\sum_{j=1}^N
(s_{bj}\widehat\zeta_{jl}-\widehat P_l(F_b)/N)^2}
=\frac{(N-1)E_l(F_b)^2}{N-E_l(F_b)^2}.
\end{equation}
The increasing relationship in \eqref{eq:ks_monotone},
with value $+\infty$ at the upper endpoint, implies identical
marginal ranks and $p$-values under a common tie convention.
Since all outcomes use the same $N$ clusters, this relationship
also preserves the ranks of the synchronized max statistics. We next compare the effective-score test with \emph{clip}.
\begin{proposition}\label{prop:pathwise}
Let $p^S$ and $p^P$ denote the $p$-values of \emph{clip} and the corresponding clusterwise effective-score test for $H_0$,
respectively, computed using the same transformations and coordinatewise non-decreasing combining function $\psi$. If $a_l(F_b)>0$ and $\widehat{\sigma}_l>0$
for all $l\in L$ and $b=1,\ldots,B$, then $p^S\geq p^P$ for every realization of the data and transformations. Consequently, whenever \emph{clip} rejects, the effective-score test also rejects at the same level.
\end{proposition}

\section{Rejection probabilities and convergence rates}
\label{sec:quantitative}

In this section, $B$ and $|L|$ are fixed, and 
$\alpha_B=\lfloor\alpha B\rfloor/B$ is the attainable resampling level. We also focus on the max-$T$  \citep{westfall1993resampling} combining function.

\subsection{Gaussian approximation}
\label{sec:quantitative_gaussian}

The first source of approximation error is the Gaussian approximation to the joint oracle resampling distribution. The following third-moment condition controls this approximation and makes explicit the role of cluster heterogeneity.

\begin{assumption} \label{ass:quantitative} For all sufficiently large $N$, $\Omega_N$ is positive definite and \begin{equation} \label{eq:rhoN} \rho_N = \sum_{j=1}^N \operatorname{E} \left[ \left\| \Omega_N^{-1/2}\boldsymbol{\zeta}_j^\ast \right\|_2^3 \right] \longrightarrow 0. \end{equation} \end{assumption}

$\rho_N$ in~\eqref{eq:rhoN} quantifies the Gaussian approximation error,
with $\rho_N\to0$ ensuring that no single cluster dominates asymptotically
\citep[Theorem~1.1]{raic2019multivariate}.

For $l\in L$, define the oracle standardized transformed score $S_l^\ast(F_b) = \frac{1}{v_{Nl}} \sum_{j=1}^N s_{bj}\zeta_{jl}^\ast$ with  $v_{Nl}^2=[\Omega_N]_{ll}$, and let $\Delta_N = \max_{\substack{1\leq b\leq B\\ l\in L}} \left| \widehat S_l(F_b)-S_l^\ast(F_b) \right|$. 

\begin{theorem} \label{thm:quantitative-level} 
Under $H_0$, suppose Assumptions~\ref{ass1} and~\ref{ass:quantitative} hold, and that the oracle scores use fixed or design-measurable working weights. Then there exist constants $C^{\mathrm{BE}},C^{\mathrm{ac}}<\infty$, depending only on $B$ and $|L|$, such that, for every $\delta>0$, \begin{equation} \label{eq:quantitative-bound} \left| \Pr(p\leq\alpha)-\alpha_B \right| \leq C^{\mathrm{BE}}\rho_N + C^{\mathrm{ac}}\delta + \Pr(\Delta_N>\delta). \end{equation} \end{theorem}

The bound in~\eqref{eq:quantitative-bound} separates the Gaussian
approximation error, controlled by $\rho_N$, from the discrepancy between the estimated and oracle statistics. The term $C^{\mathrm{ac}}\delta$ controls the probability mass near the rejection boundary.

Under Assumptions~\ref{ass:nuisance_rate}--\ref{ass:standardization},
$\Delta_N=o_p(1)$; estimated working weights additionally require Assumption~\ref{ass:estimated_weights}.

The overall convergence rate also depends on how quickly $\Delta_N$ vanishes. If clusters contribute comparable shares of the total variance and their standardized third moments are uniformly bounded, $\rho_N=O(N^{-1/2})$. Unequal cluster sizes can slow this rate when a few clusters account for a large share of the total variability.

\subsection{Effect of standardization under Gaussian errors}
\label{sec:quantitative_standardization}
\label{sec:gaussian-rate}

The clusterwise effective-score and \emph{clip} statistics satisfy
\begin{equation}
\label{eq:effective-standardized-relation}
E_l(F_b)=\sqrt{1-\lambda_{bl,N}}\,\widehat S_l(F_b),
\qquad
\lambda_{bl,N}=\frac{\|H_lF_bA_l\|_2^2}{a_l^2}.
\end{equation}
The quantity $\lambda_{bl,N}$ in
\eqref{eq:effective-standardized-relation} measures the loss of scale
from nuisance projection after sign flipping. Since $\lambda_{1l,N}=0$,
this loss affects the transformed scores but not the observed score.
For a marginal test, $\widehat\sigma_l$ cancels from the ranks, allowing
us to isolate the effect of $a_l(F_b)$.

For the remainder of this subsection, suppress the outcome index and,
under the marginal null, assume
$\varepsilon=y-\mu^\ast\sim N_n(0,\sigma^2W^{-1})$,
where $W$ is fixed and $\sigma^2>0$.
Let $u=A/a$ and $e=W^{1/2}\varepsilon/\sigma$; then
$\|u\|_2=1$, $Hu=0$ and $e\sim N_n(0,I_n)$.
Writing $\lambda_b=\lambda_{bl,N}$, define
\begin{equation}
\label{eq:direct-projected}
Z_b^P=u^\top F_b(I-H)e,
\qquad
Z_b^S=\frac{Z_b^P}{\sqrt{1-\lambda_b}}.
\end{equation}
The statistics in \eqref{eq:direct-projected} have the marginal ranks
of the clusterwise effective-score test and \emph{clip}, respectively.
Conditional on the transformations, $Z_b^P$ has variance $1-\lambda_b$,
whereas each non-degenerate $Z_b^S$ is standard normal.
The following conditions also control their joint dependence.

Partition $u$ by cluster and put
$\pi_j=\|u_j\|_2^2$ and
$\ell_N=\sum_{j=1}^N u_j^\top H_{jj}u_j
=\operatorname E_F\|HFu\|_2^2$.
Here $\pi_j$ is the share of oracle score variance contributed by cluster $j$, and $\ell_N$ is the expected loss from nuisance projection.

\begin{assumption}
\label{ass:direct-gaussian}
For a constant $C<\infty$ independent of $N$,
\begin{equation}
\label{eq:regular-cluster-design}
\max_j\pi_j\leq C/N,\qquad \ell_N\leq C/N.
\end{equation}
\end{assumption}

\begin{theorem}
\label{thm:gaussian-rate}
\label{cor:strict-rate-improvement-gaussian}
Under the Gaussian model above and Assumption~\ref{ass:direct-gaussian},
let $p^S$ and $p^P$ be the two-sided resempling $p$-values based on $|Z_b^S|$ and $|Z_b^P|$, respectively.
If any sampled transformation has $\lambda_b=1$, set $p^S=1$.
For $1\leq\lfloor\alpha B\rfloor\leq B-1$,
\begin{align}
\Pr(p^S\leq\alpha)&=\alpha_B+O(N^{-3/2}),\nonumber\\
\Pr(p^P\leq\alpha)&=\alpha_B+
C_{B,\alpha}(B-1)\ell_N+O(N^{-3/2}),
\label{eq:new-rate-comparison}
\end{align}
where $C_{B,\alpha}>0$. If $\ell_N\asymp N^{-1}$, the clusterwise effective-score test has positive size error of order $N^{-1}$, whereas \emph{clip} has
absolute size error $O(N^{-3/2})$.
%The expansion for $p^P$ also holds for the matched nonstudentized restricted wild cluster bootstrap and the equivalent Kline--Santos LM score bootstrap.
\end{theorem}

By the equivalences in Section~\ref{sec:ks_equivalence}, the expansion
for $p^P$ also applies to the corresponding restricted wild (cluster) bootstrap and null-restricted scalar LM score bootstrap.
A comparison with studentized wild cluster bootstraps such as WCR-S \citep{mackinnon2023fast} requires an additional step. Their standard errors are estimated from each bootstrap sample and depend on the same data as the score numerator. Their calibration therefore requires jointly expanding both quantities, as in the Edgeworth approach of \citet{djogbenou2019asymptotic}. 

\section{Simulations}\label{sec:simulations}

Across all simulation scenarios, we compare the proposed \emph{clip} procedure, introduced in Section~\ref{sec:clustertest}, with standard benchmark methods. We consider four implementations of \emph{clip}, i.e., under different working weight matrices. Specifically, we consider: (i) a working-independence version, denoted ``clip $(0\mid \text{Cluster})$'' in the plots, i.e., $\widehat W_{jl}=I$; (ii) a version based on the working covariance implied by the full LMM in~\eqref{eq:model} (i.e., the correct specification), denoted ``clip $ X + Z +(1+X+Z\mid \text{Cluster})$''; (iii) a version based on a random-intercept working model, denoted ``clip $X+ Z + (1\mid \text{Cluster})$''; and (iv) an oracle version based on the true covariance, denoted ``clip true''. In cases (ii) and (iii), the working covariance matrices are estimated using the \texttt{lme4} \texttt{R} package \citep{bates2015fitting} under the alternative hypothesis. Results using working covariance matrices estimated under the null \eqref{eq:null} with $\beta_l^0 = 0$ are presented in the supplementary materials. We estabilish the size-error improvement over the null-restricted scalar LM score bootstrap theoretically in Section~\ref{sec:gaussian-rate}. We assess whether this improvement extends to modified-score variants through simulations, focusing on WCR-S following the recommendations of \citet{mackinnon2023fast}.

In all scenarios, we perform $1000$ simulations, use a nominal significance level $\alpha = 0.05$, and generate $B = 1000$ clusterwise sign-flip transformations.

\subsection{Univariate Scenario}\label{subsec:sim_uni}

We compare the \emph{clip} approach with LMM, the restricted wild cluster bootstrap (WCR-S) \citep{mackinnon2023fast}, implemented with both Rademacher (WCB Rademacher in the plots) and Webb auxiliary weights (WCB Webb) \citep{fwildclusterboot}, linear models using the usual OLS variance estimator (LM) and the heteroskedasticity-consistent version (LM HC3) \citep{mackinnon1985some}, and GEE with identity working correlation matrix in terms of type I error control and power when $M=1$.

The data are generated from model~\eqref{eq:model} with a single nuisance covariate $z_{ij}$. The covariates $x_{ij}$ and $z_{ij}$ are drawn from a bivariate standard normal distribution with correlation $0.7$. The cluster-specific random effects include a random intercept and random slopes for $x_{ij}$ and $z_{ij}$ and follow a multivariate normal distribution with mean zero, unit marginal variances, and a common pairwise correlation of $0.5$. The inferential target is the fixed effect of $x_{ij}$. We test $H_0:\beta=0$ versus $H_1:\beta \ne 0$, setting $\beta=0$ to evaluate Type~I error control and $\beta=0.5$ to assess power, while fixing $\gamma=1$ in both scenarios.

\begin{figure}
\centering
\includegraphics[width=\linewidth]{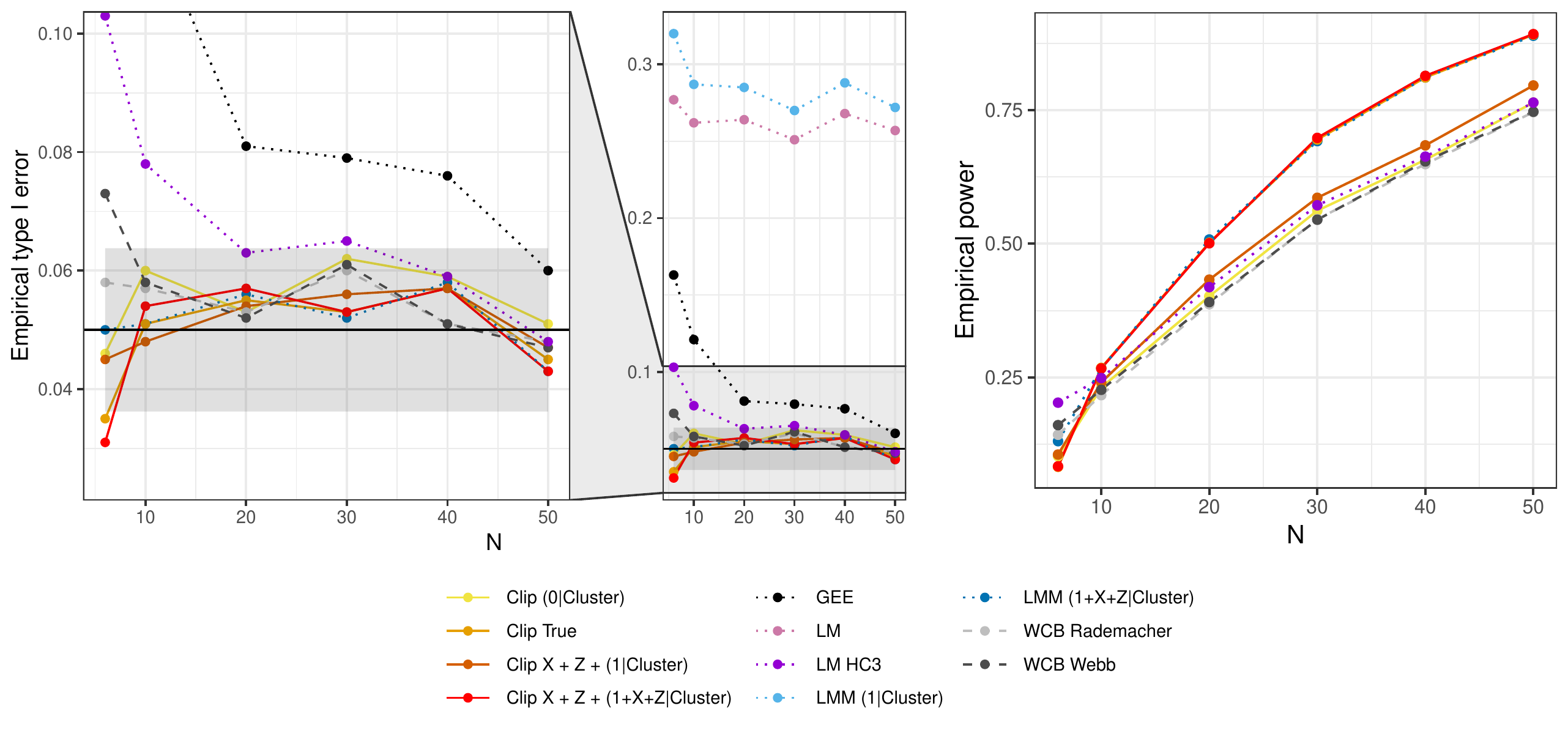}
\caption{Estimated type I error (left) and power (right) for $N\in\{6,10,20,30,40,50\}$ clusters
and $n_j\sim\text{Uniform}(10,20)$ repeated measurements.
Each line represents one model. In the type I error panels, the grey band indicates the $95\%$ confidence interval under
$\alpha=0.05$, marked by the horizontal black line.}
\label{fig:error_power}
\end{figure}

Figure \ref{fig:error_power} shows the type I error (left panel) and power (right panel) for $N \in \{6, 10, 20,30,40,50\}$ with $n_j \sim \text{Uniform}(10,20)$. The performance of the LMM is highly sensitive to the correct specification of the random-effects structure, whereas HC3 is anti-conservative in some small-sample scenarios. GEE systematically fails to achieve the nominal type~I error level, even for moderate to large values of $N$. The proposed method, instead, consistently controls the type~I error while attaining empirical power comparable to that of the correctly specified LMM. Notably, in the \emph{clip} approach, variance misspecification primarily affects power rather than type~I error control, unlike the LMM. As expected, power is higher when the true covariance structure is used or when the covariance model is correctly specified and estimated, notably surpassing the WCB.

Additional simulation results for highly unbalanced cluster sizes, non-Gaussian random effects, and settings with a small number of clusters and large within-cluster sample sizes are reported in the supplementary materials.

\subsection{Multivariate Scenario}\label{subsec:sim_multi}

We assess FWER control and power for our method relative to the LMM, WCBs, and LM HC3 in the case of multivariate responses. Multiple testing is addressed via the Holm--Bonferroni correction \citep{holm1979simple} for the LMM and LM HC3, and via the max-$T$ approach \citep{westfall1993resampling} for the \emph{clip} and WCB methods \citep{umlauft2019wild}.

We simulate data according to Equation~\eqref{eq:model} with $M = 10$, assuming an equicorrelation structure among the dependent variables with correlation $0.4$ and standard deviations in the set $\{1, 2, 3, 4, 5\}$. The fixed-effects and random-effects specification follows Subsection~\ref{subsec:sim_uni}, with the same random-effects structure used for each outcome. We consider $N = 100$ clusters and $n_j \sim \text{Uniform}\{20, 50\}$. We evaluate FWER over the $8$ true null outcomes by setting $\beta_1,\ldots,\beta_{8}=0$, and assess power over the remaining $2$ outcomes by setting $\beta_{9} = \beta_{10}=0.2$.

Figure \ref{fig:error_power_multi} shows the estimated FWER (left panel) and power (right panel). The LMM and LM HC3 results under Holm-Bonferroni correction are strongly conservative, and this conservativeness increases with higher observed correlation between the dependent variables. In contrast, the proposed method and WCB achieve the nominal FWER level of $0.05$ and maintain good power properties, as shown in the right panel of Figure \ref{fig:error_power_multi}. As in the univariate case, we gain power when a true covariance matrix is used or properly fitted.

\begin{figure}
\centering
\includegraphics[width=.9\linewidth]{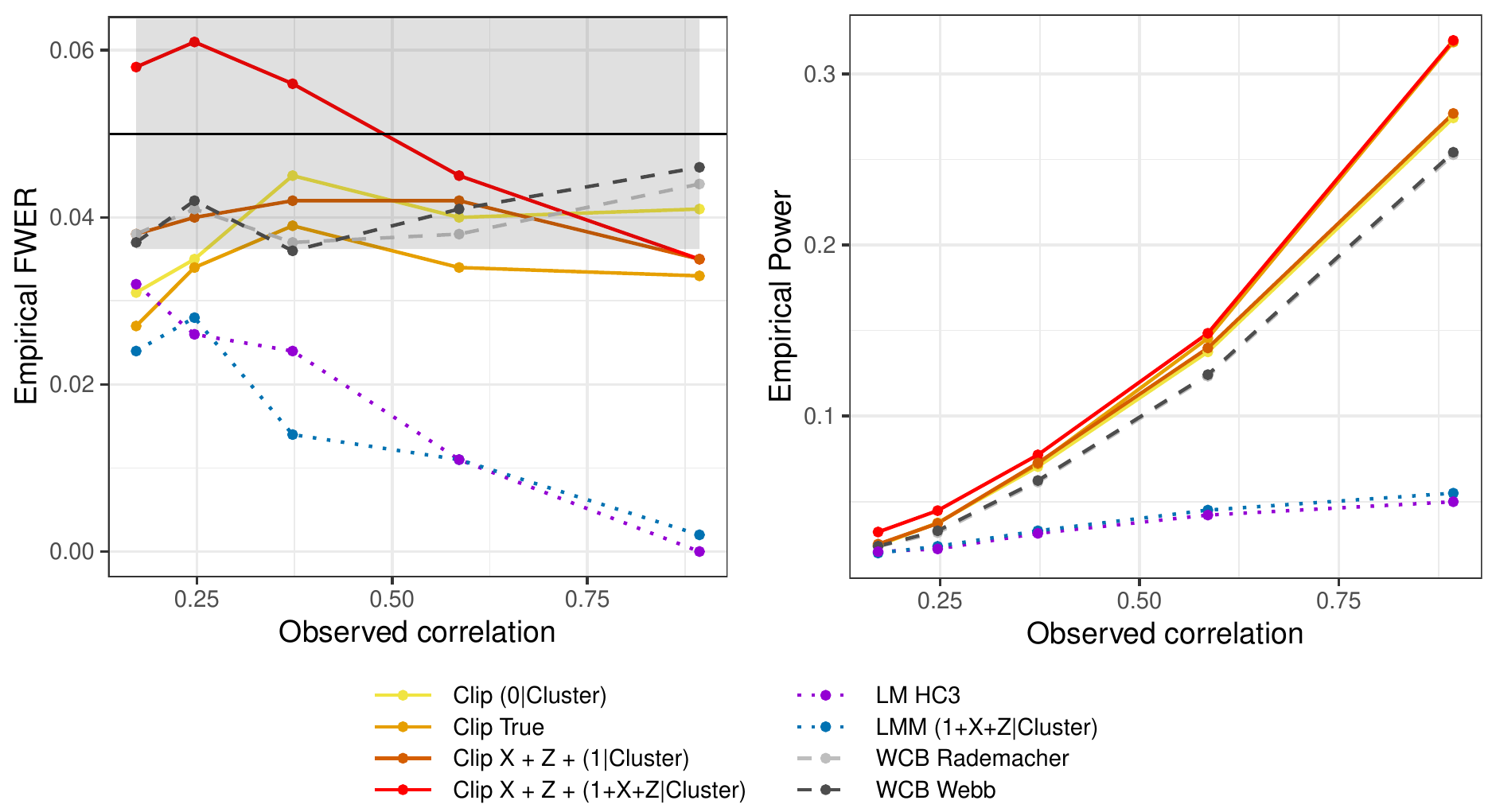}
\caption{Estimated FWER (left) and power (right) for $M=10$ outcomes,
$N=100$ clusters, and $n_j\sim\text{Uniform}(20,50)$ repeated
measurements. Each line represents one model. In the left panel, the grey band indicates the $95\%$ confidence interval under
$\mathrm{FWER}=0.05$, marked by the horizontal grey line.}
\label{fig:error_power_multi}
\end{figure}

\section{Real data analysis}\label{sec:EEG}

We analyze EEG data from the ERP CORE (Event-Related Potential Compendium of Open Resources and Experiments) project \citep{kappenman2021erp}. The experiment includes $37$ participants, with neural activity recorded from $28$ scalp electrodes during a face-perception task. On each trial, an image from one of four conditions (intact faces, intact cars, scrambled faces, scrambled cars) was presented for $300$~ms (inter-stimulus interval: $1100$-$1300$~ms), and the ERP was recorded. For further details, see \cite{kappenman2021erp,rossion2011erp}. The data are publicly available at \url{https://osf.io/thsqg/wiki/home/}.

For participant $j=1,\ldots,N$, let $n_j$ denote the number of trials retained after preprocessing. We test the effects of a $2\times 2$ within-cluster (i.e., within-subject) design, i.e., category (faces versus cars) and scramble (scrambled versus intact) and their interaction, on ERP amplitudes. On average, $n_j \approx 270$. For every cluster/subject, we have a $\mathbf{Y}_j = [y_{ijm}]$ matrix with dimension $n_j \times 28$.

We apply \emph{clip} and compare it with two LMMs having different random structures, the LM HC3 and WCB approaches. In all cases, for each electrode $m$, the fixed-effects mean model was $\operatorname{E}(y_{ijm}\mid C_{ij},S_{ij}) = \beta_{0m}+\beta_{Cm}C_{ij}+\beta_{Sm}S_{ij} +\beta_{CS,m}C_{ij}S_{ij}$, where $C_{ij}$ and $S_{ij}$ denote the coded category and scramble indicators, respectively. The random specification of the LMMs (also used for the \emph{clip} procedure) is shown in parentheses in Figure \ref{fig:eeg}. Separate working covariance matrices were estimated for each electrode.  We separately tested the category, scramble, and category-by-scramble interaction coefficients at each electrode, controlling the FWER over the 28 electrodes using the Holm-Bonferroni procedure \citep{holm1979simple} for the LMMs and LM HC3 inferences, and the max-$T$ approach for the \emph{clip} and WCBs. Here, we consider the Rademacher weights for WCB.

Figure~\ref{fig:eeg} shows the adjusted $p$-values for the category factor (results for the other effects are in the supplementary materials), with non-significant electrodes at $\alpha = 0.05$ shown in grey. Both the fully specified and the simpler LMM identify significant electrodes, but the resulting patterns differ substantially. This instability aligns with our simulation study, which found that LMM-based inference can be highly sensitive to the random-effects specification, directly affecting type~I error. The LM HC3 approach yields an intermediate pattern of discoveries; however, the simulation study showed that type~I error control is not guaranteed in finite samples, so these findings should be interpreted with caution. Finally, the WCB finds discoveries, although it identifies fewer electrodes than several \emph{clip} versions.

By contrast, \emph{clip} does not require specifying the random-effects structure, but its efficiency may suffer when an LMM-derived working covariance is poorly estimated. In Figure~\ref{fig:eeg}, the more complex
specification (top-right panel) yields fewer discoveries than the simpler alternatives (top-left and top-middle panels). This may reflect unstable
covariance estimation: the corresponding full LMM exhibits convergence problems at 10 of the 28 electrodes, whereas the simpler models do not. A richer working covariance can introduce estimation noise rather than improve efficiency.

The choice of the working covariance must be made before inspecting the inferential results to avoid cherry-picking. For LMM-based weights, we recommend comparing a small set of plausible models using fit stability and convergence diagnostics. If none provides a stable covariance estimate, working independence is often preferable.

\begin{figure}[h] \centering \includegraphics[width=\textwidth]{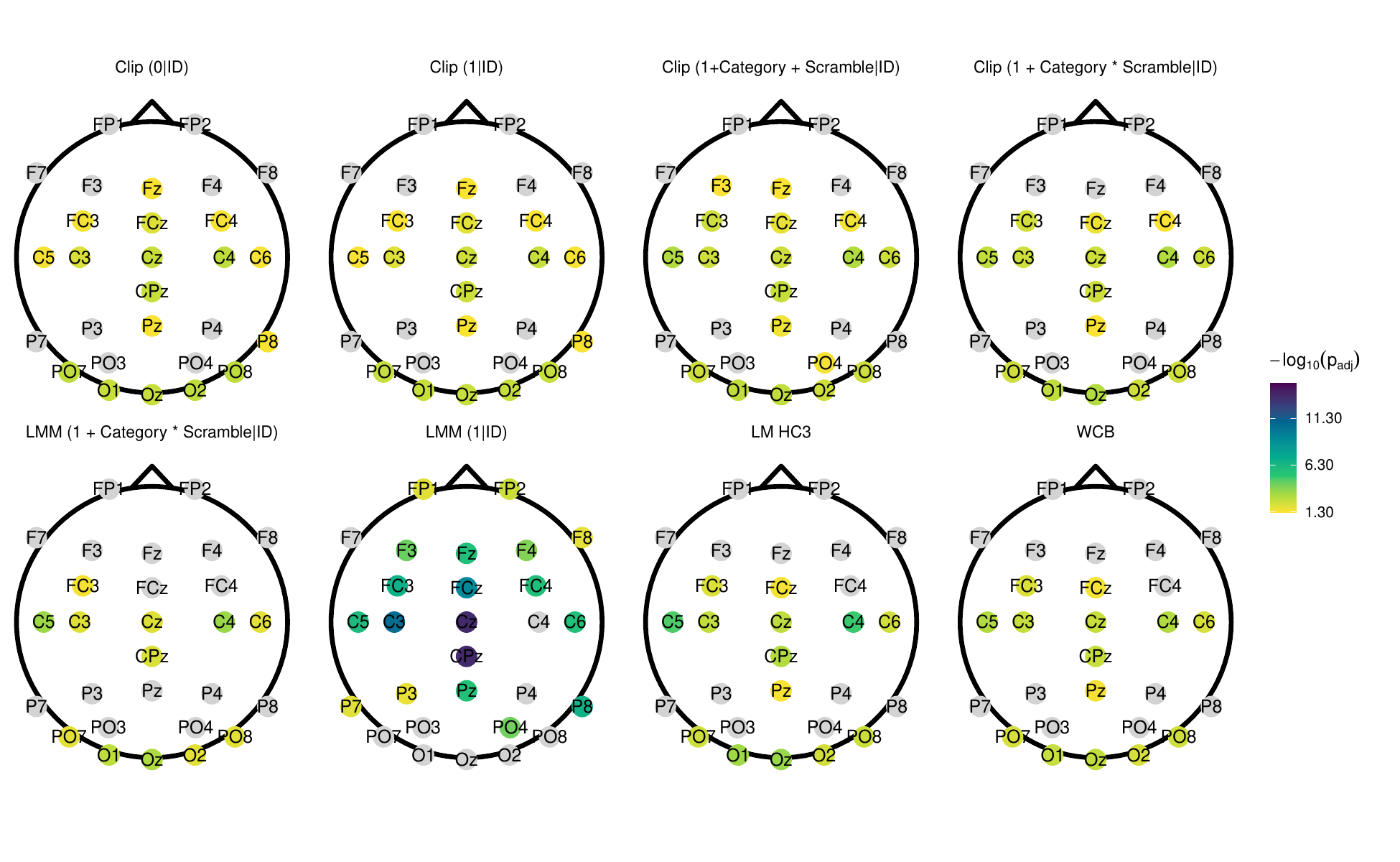} \caption{Adjusted $p$-values for the category effect for each electrode, using \emph{clip} (top row) and LMM/LM HC3/WCB methods (bottom row). The corresponding working covariance matrix and random-effects structure are in parentheses.} \label{fig:eeg} \end{figure}

\section{Discussion}\label{sec:discussion}

In \emph{clip}, the working covariance can improve efficiency without requiring correct specification for first-order validity. Synchronized clusterwise transformations support global and simultaneous inference across responses without estimating a between-response covariance matrix. We also identify a Gaussian setting in which transformation-specific standardization reduces the size error from order $N^{-1}$ for the corresponding nonstandardized wild (cluster) bootstrap procedures to $O(N^{-3/2})$.

As in many methods for clustered data, validity requires a correctly specified marginal mean and independent clusterwise contributions, making the choice of clusters essential. More complex dependence structures remain a direction for further research, although \emph{clip} can accommodate some crossed designs when conditioning restores independence while preserving the coefficient being tested (see the supplementary material).

Future work could extend standardization and the rate comparison to modified-score wild cluster bootstraps such as WCR-S \citep{mackinnon2023fast}. Another open question concerns validity when the number of tested outcomes or transformations increases with the number of clusters. Extensions to nonlinear marginal mean models are also of interest.

\section*{Competing interests}
No competing interests are declared.

\section*{Author contributions statement}
A.A: Conceptualization, Formal Analysis, Methodology, Software, Visualization, Writing – Original Draft Preparation. L.F: Conceptualization, Methodology, Software, Writing - Review \& Editing, Software, Supervision

\section*{Acknowledgments}
We thank Dr. Antonio Maffei for the preprocessed EEG data used in Section \ref{sec:EEG}.

% ----------------------------------------------------------------------
% Supplementary material: independent printed numbers and hyperlink anchors.
% ----------------------------------------------------------------------
\clearpage
\section*{Supplementary materials}
\addcontentsline{toc}{section}{Supplementary materials}

\setcounter{section}{0}
\setcounter{subsection}{0}
\setcounter{subsubsection}{0}
\setcounter{theorem}{0}
\setcounter{lemma}{0}
\setcounter{assumption}{0}
\setcounter{definition}{0}
\setcounter{example}{0}
\setcounter{remark}{0}
\setcounter{figure}{0}
\setcounter{table}{0}
\setcounter{equation}{0}
\renewcommand{\thesection}{S\arabic{section}}
\renewcommand{\thetheorem}{S\arabic{theorem}}
\renewcommand{\thelemma}{S\arabic{lemma}}
\renewcommand{\theassumption}{S\arabic{assumption}}
\renewcommand{\thedefinition}{S\arabic{definition}}
\renewcommand{\theexample}{S\arabic{example}}
\renewcommand{\theremark}{S\arabic{remark}}
\renewcommand{\thefigure}{S\arabic{figure}}
\renewcommand{\thetable}{S\arabic{table}}
\numberwithin{equation}{section}

% Distinguish supplementary PDF destinations from the main text's destinations.
\def\theHsection{supp.\arabic{section}}
\def\theHsubsection{\theHsection.\arabic{subsection}}
\def\theHsubsubsection{\theHsubsection.\arabic{subsubsection}}
\def\theHtheorem{supp.\arabic{theorem}}
\def\theHproposition{\theHtheorem}
\def\theHcorollary{\theHtheorem}
\def\theHlemma{supp.\arabic{lemma}}
\def\theHassumption{supp.\arabic{assumption}}
\def\theHdefinition{supp.\arabic{definition}}
\def\theHexample{supp.\arabic{example}}
\def\theHremark{supp.\arabic{remark}}
\def\theHfigure{supp.\arabic{figure}}
\def\theHtable{supp.\arabic{table}}
\def\theHequation{\theHsection.\arabic{equation}}

\section{Proofs of the results}\label{supp-appA}

\subsection{Proof of Lemma~\ref{lem:exchangeability}}

The argument follows the results of \citet[Theorem~1]{hemerik2020robust} and its multivariate extension in \citet[Lemma~3.1]{de2024permutation}, with independent clusters replacing individual observations.

Let $m=|L|$, $X_{j,N}=\Lambda_N^{-1/2}\boldsymbol\zeta_j^\ast$, and let $M_N^\ast$ denote the oracle version of $M_N$, with elements $S_l^\ast(F_b)=v_{Nl}^{-1}\sum_j s_{bj}\zeta_{jl}^\ast$. The vector obtained by stacking the rows of $M_N^\ast$ is the sum of the independent, centered vectors \[ \Xi_{j,N} = (s_{1j},\ldots,s_{Bj})^\top\otimes X_{j,N}. \] Since $s_{1j}=1$ and the remaining signs are independent with mean zero, \[ \sum_{j=1}^N\operatorname{Cov}(\Xi_{j,N}) = I_B\otimes R_N, \qquad \sum_{j=1}^N\operatorname{E}\|\Xi_{j,N}\|_2^d = B^{d/2}\sum_{j=1}^N\operatorname{E}\|X_{j,N}\|_2^d \longrightarrow0. \] Assumptions~\ref{ass:lyapunov} and~\ref{ass:correlation} and the multivariate Lyapunov central limit theorem therefore give \[ M_N^\ast\xrightarrow{d}M, \] where the rows of $M$ are independent $N_m(0,R)$ random vectors. The conclusion also holds for singular $R$ by the Cram\'er--Wold theorem \citep{cramer1936some}.

We now replace the oracle contributions with their estimated counterparts. Let $e_{Nl}=\widehat{\boldsymbol\gamma}_l-\boldsymbol\gamma_l^\ast$ and $C_{Nl}^2=\sum_j\|c_{jl}\|_2^2$. From $\widehat\zeta_{jl}-\zeta_{jl}^\ast=-c_{jl}^\top e_{Nl}$, \[ \frac{1}{v_{Nl}} \left| \sum_j s_{bj}(\widehat\zeta_{jl}-\zeta_{jl}^\ast) \right| \leq \frac{\|e_{Nl}\|_2}{v_{Nl}} \left\|\sum_j s_{bj}c_{jl}\right\|_2. \] For $b=1$, the left-hand side is zero because $\sum_jc_{jl}=0$. For $b\geq2$, $\operatorname{E}_F\|\sum_j s_{bj}c_{jl}\|_2^2=C_{Nl}^2$. Consequently, Assumptions~\ref{ass:nuisance_rate} and \ref{ass:nuisance_sensitivity} imply that the preceding discrepancy is $o_p(1)$ uniformly over the fixed numbers of outcomes and transformations.

Finally, let $d_{bl,N}=v_{Nl}/\{a_l(F_b)\widehat\sigma_l\}$. Then \[ \widehat S_l(F_b) = d_{bl,N} \left\{ S_l^\ast(F_b) + \frac{1}{v_{Nl}} \sum_j s_{bj}(\widehat\zeta_{jl}-\zeta_{jl}^\ast) \right\}. \] By Assumption~\ref{ass:standardization}, $\max_{b,l}|d_{bl,N}-1|=o_p(1)$, while the oracle central limit theorem implies $\max_{b,l}|S_l^\ast(F_b)|=O_p(1)$. Hence $\max_{b,l}|\widehat S_l(F_b)-S_l^\ast(F_b)|=o_p(1)$. The result follows from Slutsky's theorem. \qed

\subsection{Proof of Theorem~\ref{thm:clip_validity}}

Let $g(z)=\psi(|z_l|:l\in L)$. By Lemma~\ref{lem:exchangeability} and the continuous mapping theorem, \[ \bigl(\widehat T(F_1),\ldots,\widehat T(F_B)\bigr) \xrightarrow{d} \bigl(g(Z_1),\ldots,g(Z_B)\bigr), \] where $Z_1,\ldots,Z_B$ are independent $N_{|L|}(0,R)$ random vectors. By assumption, the limiting statistics are independent, identically distributed, and continuous. Lemma 1 of \cite{hemerik2020robust} therefore gives \[ \Pr(p\leq\alpha) \longrightarrow \frac{\lfloor\alpha B\rfloor}{B} \leq\alpha. \] See also \citet[Theorem~3.2]{de2024permutation}. \qed

\subsection{Proof of Corollary~\ref{cor:estimated_weights}}

Let $M_N^W$ and $M_N^{\widehat W}$ denote the transformed score matrices computed using $W_l$ and $\widehat W_l$, respectively. Assumption~\ref{ass:estimated_weights} gives $\|M_N^{\widehat W}-M_N^W\|_{\max}=o_p(1)$. Since Lemma~\ref{lem:exchangeability} gives $M_N^W\xrightarrow{d}M$, Slutsky's theorem implies $M_N^{\widehat W}\xrightarrow{d}M$ with the same limiting distribution. The conclusion now follows from Theorem~\ref{thm:clip_validity}. \qed

\subsection{Proof of Proposition~\ref{prop:matched-wcb}.}
\label{supp-app:exact-wcb}

Suppress the outcome index and let $r_0=y-\widehat\mu$ denote the null-restricted residual vector. The Frisch--Waugh--Lovell identity and $A^\top W^{1/2}Z=0$ give \[ \widehat\beta^*(F)-\beta^0 =\frac{A^\top W^{1/2}Fr_0}{a^2} =\frac{\widehat P(F)}{a^2}. \] For $F=I$, $y^*(I)=y$, giving the observed coefficient. The positive factor $a^{-2}$ is common to all transformations, so the ranks and $p$-values agree.  \qed

\subsection{Proof of Proposition~\ref{prop:pathwise}}
\label{supp-app:pathwise}

Orthogonality of the projection and $F_b^\top F_b=I$ give \[ a_l(F_b)^2 =\|(I-H_l)F_bA_l\|_2^2 =a_l^2-\|H_lF_bA_l\|_2^2 \leq a_l^2. \] Thus $|\widehat S_l(F_b)|\geq |E_l(F_b)|$ for every coordinate, with equality for $b=1$. Since $\psi$ is non-decreasing, every transformation counted in $p^P$ is also counted in $p^S$. Hence $p^S\geq p^P$. \qed

\subsection{Proof of Theorem~\ref{thm:quantitative-level}}

Set $k=B-\lfloor\alpha B\rfloor$. For $t=(t_1,\ldots,t_B)^\top$, let $t_{(k)}$ denote its $k$th order statistic. Then \begin{equation} \label{supp-eq:app-rank-identity} \left\{ B^{-1}\sum_{b=1}^B\mathbf{1}\{t_b\geq t_1\}\leq\alpha \right\} = \{t_1>t_{(k)}\} = \left\{ \sum_{b=2}^B\mathbf{1}\{t_b<t_1\}\geq k \right\}. \end{equation} This identity remains valid in the presence of ties; when $k=B$, all three events are empty.

Write $m=|L|$ and let $Y_{j,N}=\Omega_N^{-1/2}\boldsymbol\zeta_j^\ast$. Define $\mathcal X_{j,N}=(s_{1j},\ldots,s_{Bj})^\top\otimes Y_{j,N}$, where $s_{1j}=1$. Under the null, the vectors $\mathcal X_{j,N}$ are independent and centered. Independence of the random signs gives \[ \sum_{j=1}^N\operatorname{Cov}(\mathcal X_{j,N})=I_{Bm}, \qquad \sum_{j=1}^N\operatorname{E}\|\mathcal X_{j,N}\|_2^3 =B^{3/2}\rho_N. \] The multivariate Berry--Esseen bound of \citet[Theorem~1.1]{raic2019multivariate} therefore approximates the joint oracle score vector by a standard Gaussian vector in convex-set distance, with error at most a constant times $\rho_N$. Applying the blockwise linear map $\Lambda_N^{-1/2}\Omega_N^{1/2}$ gives a Gaussian reference consisting of independent vectors $Z_{N1},\ldots,Z_{NB}\sim N_m(0,R_N)$.

Let $T_b^\ast=\max_{l\in L}|S_l^\ast(F_b)|$ and $T_b^G=\max_{l\in L}|Z_{Nb,l}|$. After partitioning by coordinate signs, maximizing coordinates, and rank orderings, the rank event in~\eqref{supp-eq:app-rank-identity} is a finite disjoint union of convex polyhedral regions, with a number of regions depending only on the fixed values of $B$ and $m$. Applying the preceding Berry--Esseen bound to these regions gives \[ \left| \Pr\{T_1^\ast>T_{(k)}^\ast\} - \Pr\{T_1^G>T_{(k)}^G\} \right| \leq C^{\mathrm{BE}}\rho_N. \] The Gaussian statistics are independent and identically distributed and have no ties. Their ranks are consequently uniform, and the second probability is $\alpha_B$.

It remains to replace the oracle scores by their feasible counterparts. The Gaussian anti-concentration inequality for maxima \citep{chernozhukov2015comparison}, applied to the maximum of $(Z_{Nb,l},-Z_{Nb,l}:l\in L)$, yields \[ \Pr\left\{ \min_{2\leq b\leq B}|T_1^\ast-T_b^\ast|\leq2\delta \right\} \leq C^{\mathrm{BE}}\rho_N+C^{\mathrm{ac}}\delta. \] Indeed, the corresponding Gaussian probability is $O(\delta)$, and the oracle probability differs from it by $O(\rho_N)$ by the same finite-region argument used above.

The max-$T$ function is one-Lipschitz with respect to the maximum norm. Thus, on $\{\Delta_N\leq\delta\}$, $\max_b|\widehat T(F_b)-T_b^\ast|\leq\delta$. The feasible and oracle rank decisions can then differ only on the near-tie event in the preceding display. Combining these observations with \eqref{supp-eq:app-rank-identity} proves~\eqref{eq:quantitative-bound}. \qed

\subsection{Gaussian rank expansion}
\label{supp-app:direct-gaussian}

We first expand the Gaussian rejection probability conditional on the transformations, then average over the signs. The covariance in this expansion includes the dependence induced by nuisance estimation.

With $k$ as in \eqref{supp-eq:app-rank-identity}, let \[ \mathcal A=\left\{z\in\mathbb R^B: \sum_{b=2}^B\mathbf1\{|z_b|<|z_1|\}\geq k\right\}. \] This is the rejection region for the absolute statistics. It is unchanged by reversing any coordinate sign. In the following lemma, $1\leq k\leq B-1$.

\begin{lemma} \label{supp-lem:gaussian-rank-expansion} Let $Z\sim N_B(0,I_B)$ and define \[ A_{bc}=\operatorname E\{\mathbf1_{\mathcal A}(Z) (Z_b^2-1)(Z_c^2-1)\},\qquad b<c. \] Then $\sum_{b<c}A_{bc}=0$. For a correlation matrix $R=I_B+E$ sufficiently close to $I_B$, a vector $\lambda=(0,\lambda_2,\ldots,\lambda_B)$ sufficiently close to zero with nonnegative entries, and $D_\lambda=\operatorname{diag}(\sqrt{1-\lambda_b})$, let $Z_R\sim N_B(0,R)$. Uniformly over these $R$ and $\lambda$, \begin{align} \Pr(D_\lambda Z_R\in\mathcal A) &=\alpha_B+C_{B,\alpha}\sum_{b=2}^B\lambda_b +\frac12\sum_{b<c}A_{bc}E_{bc}^2 \nonumber\\ &\quad+O\left(\|E\|_{\mathrm{op}}^3+ \|\lambda\|_2^2+\|E\|_{\mathrm{op}}\|\lambda\|_2\right), \label{supp-eq:rank-density-expansion} \end{align} where, writing $G(x)=\Pr(|Z_1|\leq x)$ and $g=G'$ on $x>0$, \[ C_{B,\alpha}=\frac12{B-2\choose k-1} \int_0^\infty xg(x)^2G(x)^{k-1} \{1-G(x)\}^{B-1-k}\,dx>0. \] \end{lemma}

\begin{proof} The covariance of $D_\lambda Z_R$ remains in a compact positive-definite neighborhood of $I_B$. The Gaussian density and its first three covariance derivatives are bounded there by an integrable function $C(1+\|z\|^K)\exp(-c\|z\|^2)$. We can therefore integrate its Taylor expansion over $\mathcal A$, with a uniform remainder of the stated order.

At $\lambda=0$, the first-order terms in the correlations are $E_{bc}z_bz_c$. Their integrals vanish because $\mathcal A$ is even in every coordinate. Second-order terms involving distinct pairs also have an odd power in at least one coordinate and integrate to zero. The remaining terms are $\tfrac12 E_{bc}^2(z_b^2-1)(z_c^2-1)$, giving the quadratic term in \eqref{supp-eq:rank-density-expansion}.

To show that the coefficients sum to zero, take $R(\rho)=(1-\rho)I_B+\rho\mathbf1\mathbf1^\top$ for small $\rho$. This Gaussian distribution is exchangeable and has no ties in its absolute ranks, so $\Pr(Z_{R(\rho)}\in\mathcal A)=\alpha_B$. The coefficient of $\rho^2$ is therefore $\tfrac12\sum_{b<c}A_{bc}=0$.

For the scale derivative at $R=I_B$, condition on $|Z_1|=x$. Coordinate $b$ falls below the observed statistic with probability $G(x/\sqrt{1-\lambda_b})$, whose derivative at zero is $xg(x)/2$. The derivative of the probability of at least $k$ successes is this quantity multiplied by the probability of exactly $k-1$ successes among the other $B-2$ coordinates. Integration over $x$ gives $C_{B,\alpha}>0$, the same for every $b\geq2$. The remaining scale and mixed terms are $O(\|\lambda\|_2^2+\|E\|_{\mathrm{op}}\|\lambda\|_2)$. \end{proof}

Conditional on the transformations, whenever all $1-\lambda_b>0$, the standardized Gaussian vector in \eqref{eq:direct-projected} has correlation matrix $R=(R_{bc})$, where
\begin{equation}
\label{supp-eq:projected-correlation}
R_{bc}=\frac{u^\top F_b(I-H)F_cu}
{\sqrt{(1-\lambda_b)(1-\lambda_c)}}.
\end{equation}
Rejection probabilities below average over both the errors and the sampled transformations.

Let $\nu_N=\sum_{j=1}^N\pi_j^2$.

\begin{lemma} \label{supp-lem:geometry-moments} Under Assumption~\ref{ass:direct-gaussian}, there is an event $\mathcal E_N$, depending only on the transformations and design, with $\Pr_F(\mathcal E_N^c)=O(N^{-3})$, on which all $1-\lambda_b$ are bounded away from zero and the correlation matrix in \eqref{supp-eq:projected-correlation} satisfies $\|R-I_B\|_{\mathrm{op}}\leq1/2$. Moreover, for every $b<c$, \begin{align} \operatorname E_F[\mathbf1_{\mathcal E_N}R_{bc}^2] &=\nu_N+O(N^{-3/2}),\label{supp-eq:pair-cancellation}\\ \operatorname E_F[\mathbf1_{\mathcal E_N}\|R-I_B\|_{\mathrm{op}}^3] &=O(N^{-3/2}),\nonumber\\ \operatorname E_F\sum_{b=2}^B\lambda_b^2&=O(N^{-2}),\nonumber\\ \operatorname E_F[\mathbf1_{\mathcal E_N}\|R-I_B\|_{\mathrm{op}}\|\lambda\|_2] &=O(N^{-3/2}).\nonumber \end{align} \end{lemma}

\begin{proof} Write $h_b=HF_bu$ and $r_{bc}^0=u^\top F_bF_cu=\sum_js_{bj}s_{cj}\pi_j$ for $b<c$, with $s_{1j}=1$. For each pair, the products $s_{bj}s_{cj}$ are independent Rademacher variables across clusters. Hence \[ \operatorname E_F(r_{bc}^0)^2=\nu_N\leq C/N, \qquad \operatorname E_F|r_{bc}^0|^{2r}\leq C_r \nu_N^r \quad(r=1,2,3). \] To obtain the moment bound, expand the sum. Only terms in which each sign occurs an even number of times remain; their sum is bounded by a constant times $(\sum_j\pi_j^2)^r$.

Likewise, $h_b=\sum_js_{bj}H\mathsf E_ju$, where $\mathsf E_j$ selects cluster $j$. The even-moment bound for a Hilbert-space Rademacher sum gives \[ \operatorname E_F\|h_b\|^{2r} \leq C_r\left(\sum_j\|H\mathsf E_ju\|^2\right)^r =C_r\ell_N^r=O(N^{-r}),\qquad r=1,2,3. \] This bound also follows by expanding the inner products and applying Cauchy--Schwarz to the terms with even sign counts. Its constant is independent of the observation dimension. For $b=1$, $h_1=Hu=0$.

Choose a small fixed $\epsilon_B>0$ and let $\mathcal E_N$ be the event that all $\lambda_b=\|h_b\|^2$ and $|r_{bc}^0|$ are at most $\epsilon_B$. Markov's inequality, the moment bounds above and a union bound over fixed $B$ give $\Pr_F(\mathcal E_N^c)=O(N^{-3})$. Since \[ R_{bc}=\frac{r_{bc}^0-h_b^\top h_c} {\sqrt{(1-\lambda_b)(1-\lambda_c)}}, \] choosing $\epsilon_B$ small enough gives $\|R-I_B\|_{\mathrm{op}}\leq1/2$ and bounds the denominators away from zero. On $\mathcal E_N$, \[ |R_{bc}-r_{bc}^0| \leq C\{\sqrt{\lambda_b\lambda_c} +|r_{bc}^0|(\lambda_b+\lambda_c)\}. \] The moment bounds imply $\operatorname E_F[\mathbf1_{\mathcal E_N}(R_{bc}-r_{bc}^0)^2]=O(N^{-2})$. Cauchy--Schwarz then gives \[ \left|\operatorname E_F[\mathbf1_{\mathcal E_N} (R_{bc}^2-(r_{bc}^0)^2)]\right|=O(N^{-3/2}). \] Since $|r_{bc}^0|\leq1$, restricting its second moment to $\mathcal E_N$ changes it by $O(N^{-3})$. This proves \eqref{supp-eq:pair-cancellation}, with leading term $\nu_N$ for every pair.

On $\mathcal E_N$, $|R_{bc}|\leq C(|r_{bc}^0|+\sqrt{\lambda_b\lambda_c})$. The moment bounds and fixed $B$ yield $\operatorname E_F[\mathbf1_{\mathcal E_N}\|R-I_B\|_{\mathrm{op}}^3]=O(N^{-3/2})$ and the corresponding second-moment bound $O(N^{-1})$. Also, $\operatorname E_F\|\lambda\|_2^2=O(N^{-2})$. A final application of Cauchy--Schwarz gives the mixed bound. \end{proof}

\subsection{Proof of Theorem~\ref{thm:gaussian-rate}.} The null-restricted weighted residual is $\sigma(I-H)e$. Conditional on the transformations and on $\mathcal E_N$, the standardized vector in \eqref{eq:direct-projected} has distribution $N_B(0,R)$, with $R$ given by \eqref{supp-eq:projected-correlation}. Apply Lemma~\ref{supp-lem:gaussian-rank-expansion} with $\lambda=0$ and average over the transformations. The quadratic term is \[ \frac12\sum_{b<c}A_{bc} \operatorname E_F[\mathbf1_{\mathcal E_N}R_{bc}^2] =\frac{\nu_N}{2}\sum_{b<c}A_{bc}+O(N^{-3/2}) =O(N^{-3/2}). \] Lemma~\ref{supp-lem:geometry-moments} bounds the cubic remainder. The event $\mathcal E_N^c$ contributes at most $O(N^{-3})$ to the rejection probability. Thus the standardized test has rejection probability $\alpha_B+O(N^{-3/2})$.

By \eqref{eq:effective-standardized-relation}, $Z_b^P=\sqrt{1-\lambda_b}Z_b^S$, so the conditional effective-score vector is $D_\lambda Z_R$. Apply the full expansion. The leading quadratic correlation term again cancels. The scale remainder is $O(N^{-2})$, and the mixed and cubic remainders are $O(N^{-3/2})$. Finally, \[ \operatorname E_F\sum_{b=2}^B\lambda_b=(B-1)\ell_N. \] Restricting this expectation to $\mathcal E_N$ changes it by $O(N^{-3})$, since $0\leq\lambda_b\leq1$. This proves \eqref{eq:new-rate-comparison} and the rate improvement when $\ell_N\asymp N^{-1}$. The bootstrap comparisons follow from Proposition~\ref{prop:matched-wcb} and \eqref{eq:ks_monotone}. \qed

\section{Crossed designs}
\label{supp-sec:crossed}

Crossed random effects arise when observations are indexed by two or more non-nested grouping factors. A common example is a participant--item design, in which each participant is observed on multiple items and each item is observed for multiple participants \citep{baayen2008mixed,burki2018accounting,debruine2021understanding}.

Consider participants as clusters, indexed by $j=1,\ldots,N$, and the $M$ observed items as the outcome dimension. A crossed extension of the working representation~\eqref{eq:model} is \begin{equation} \label{supp-eq:crossed} \boldsymbol{y}_{jm} = \boldsymbol{x}_{jm}\beta_m + Z_{jm}\boldsymbol{\gamma}_m + Q_{jm}\boldsymbol{b}_{jm} + R_{jm}\boldsymbol{c}_m + \boldsymbol{\varepsilon}_{jm}, \end{equation} where $\boldsymbol{c}_m$ denotes the effect associated with item $m$, shared across participants. The participant-specific effects $\boldsymbol{b}_{jm}$ may be dependent across items or share common latent components.

Marginalizing over the item effects in~\eqref{supp-eq:crossed} generally induces dependence among the participant-level response matrices $\boldsymbol{Y}_1,\ldots,\boldsymbol{Y}_N$, so that Assumption~\ref{ass1} need not hold marginally. The multivariate formulation remains applicable, however, when the dependence generated by the outcome-indexed factor can be removed by conditioning without altering the coefficient being tested.

\begin{assumption} \label{supp-ass:crossed} Let $\mathcal C$ denote the sigma-field generated by the effects associated with the factor treated as the outcome dimension. For a fixed non-empty $L\subseteq\{1,\ldots,M\}$, suppose that, conditional on $\mathcal C$ and the observed design:

\begin{enumerate} \item the cluster-level response matrices $\boldsymbol{Y}_1,\ldots,\boldsymbol{Y}_N$ are independent;

\item for every $l\in L$ and $j=1,\ldots,N$, $\operatorname{E}(\boldsymbol{y}_{jl}\mid\mathcal C) = \boldsymbol{x}_{jl}\beta_l + Z_{jl}\boldsymbol{\gamma}_l(\mathcal C)$, where the tested coefficient $\beta_l$ does not depend on $\mathcal C$; \end{enumerate} \end{assumption}

Condition~(2) is the key requirement. The effects associated with the outcome-indexed factor may change outcome-specific nuisance parameters, but they must not change the coefficient being tested. For example, item-specific random intercepts satisfy this condition because, after conditioning, they can be absorbed into the item-specific nuisance intercepts. Participant-specific random intercepts, random slopes, and other cross-item dependence instead remain within each participant cluster and are therefore accommodated by the clusterwise construction. The same sign is applied to all outcome-specific score contributions from a given participant, so the synchronous transformations of Section~\ref{sec:clustertest} preserve this remaining within-cluster dependence.

\begin{proposition} \label{supp-prop:crossed} Suppose Assumption~\ref{supp-ass:crossed} holds and, for almost every realization of $\mathcal C$, the conditional versions of Assumptions~\ref{ass:lyapunov}--\ref{ass:standardization} and, when working weights are estimated, Assumption~\ref{ass:estimated_weights} hold. Suppose also that, for almost every realization of $\mathcal C$, $\psi(|Z_l|:l\in L)$ has a continuous distribution under $Z\mid\mathcal C\sim N_{|L|}(0,R(\mathcal C))$, where $R(\mathcal C)$ is the conditional limiting correlation matrix. The transformations are independent of the responses and $\mathcal C$. Then, under $H_0$, almost surely with respect to $\mathcal C$, \[ \Pr(p\leq\alpha\mid\mathcal C) \longrightarrow \frac{\lfloor\alpha B\rfloor}{B}, \] and consequently $\Pr(p\leq\alpha) \longrightarrow \frac{\lfloor\alpha B\rfloor}{B} \leq\alpha$. \end{proposition}

\begin{proof}
    Fix a realization of $\mathcal C$ outside the null set on which any of the conditional assumptions fails. Under the corresponding conditional law, Assumption~\ref{supp-ass:crossed} restores Assumption~\ref{ass1}: the clusters are independent and have conditional mean $\boldsymbol{x}_{jl}\beta_l+ Z_{jl}\boldsymbol{\gamma}_l(\mathcal C)$. The nuisance parameter is fixed under this law, whereas the tested coefficient remains $\beta_l$. Since all the remaining conditions hold conditionally by assumption, Theorem~\ref{thm:clip_validity}, or Corollary~\ref{cor:estimated_weights} when the working weights are estimated, gives \[ \Pr(p\leq\alpha\mid\mathcal C) \longrightarrow \alpha_B = \frac{\lfloor\alpha B\rfloor}{B} \qquad\text{almost surely}. \] Because the conditional rejection probabilities are bounded by one, the dominated convergence theorem yields \[ \Pr(p\leq\alpha) = \operatorname{E}\{\Pr(p\leq\alpha\mid\mathcal C)\} \longrightarrow \alpha_B \leq\alpha. \] 
\end{proof}

%The marginal conclusion follows directly from bounded convergence, since conditional rejection probabilities take values in $[0,1]$. Thus, although the clusters need not be independent marginally, the one-way multivariate representation remains valid when conditional independence restores the framework of Section~\ref{sec:framework} without changing the inferential target.

A limitation arises when the factor treated as the outcome dimension has random slopes for the tested covariate. To see this, consider one observation per participant--item pair under $y_{jm} = \alpha+\beta x_j+b_j+c_{0m}+c_{1m}x_j+\varepsilon_{jm}$, where $(c_{0m},c_{1m})^\top$ are item-specific random intercepts and slopes. Conditional on the item effects, $\operatorname{E}(y_{jm}\mid x_j,\mathcal C) = \alpha_m+(\beta+c_{1m})x_j, \qquad \alpha_m=\alpha+c_{0m}$. Hence, under the population null $H_0:\beta=0$, the conditional coefficient of the tested covariate is $c_{1m}$ rather than zero, so condition~(2) of Assumption~\ref{supp-ass:crossed} fails whenever the outcome-indexed random slope is non-degenerate. Marginalizing over the random slopes does not recover the one-way construction, because the same outcome-specific slope is shared across clusters and induces dependence among their score contributions.

More generally, the crossed-design reformulation requires one grouping factor to provide a diverging number of conditionally independent clusters, while the observed levels of the other factor form a fixed-dimensional multivariate outcome. The tested coefficient must be identifiable within each outcome, and conditioning on the effects of the outcome-indexed factor may alter outcome-specific nuisance parameters but not the coefficient being tested. The roles of the two factors may be reversed whenever these conditions hold under the alternative orientation. Random slopes for the tested covariate are therefore permitted for the cluster-indexed factor, because they contribute only to within-cluster dependence, but not for the factor treated as the outcome dimension. Inference is conditional on the observed levels of the latter factor and does not describe slope heterogeneity over new levels. If both crossed factors have random slopes for the tested covariate, neither one-way orientation satisfies Assumption~\ref{supp-ass:crossed}, and a genuinely multiway cluster-robust or resampling procedure is required \citep{cameron2011robust,mackinnon2021wild,menzel2021bootstrap}.

\section{Simulations}\label{supp-appB}

Unless otherwise stated, the data-generating mechanism, parameter values, and inferential procedures are the same as those described in Section~\ref{sec:simulations}; only the feature explicitly described in each subsection is modified.

\subsection{Working covariance estimated under the null}\label{supp-app:h0}

We repeat the simulations underlying Figure~\ref{fig:error_power}, estimating the working covariance matrices under the null hypothesis~\eqref{eq:null}, with $\beta_l^0=0$, rather than from the unrestricted model.

Figures~\ref{supp-fig:error_wc} and~\ref{supp-fig:power_wc} report the empirical type~I error and power, respectively, including both estimation approaches for comparison. The corresponding curves are nearly indistinguishable, suggesting that imposing the null when estimating the working covariance has little effect on the performance of \emph{clip} in the settings considered.

\begin{figure}[tbp] \centering \includegraphics[width=.75\linewidth]{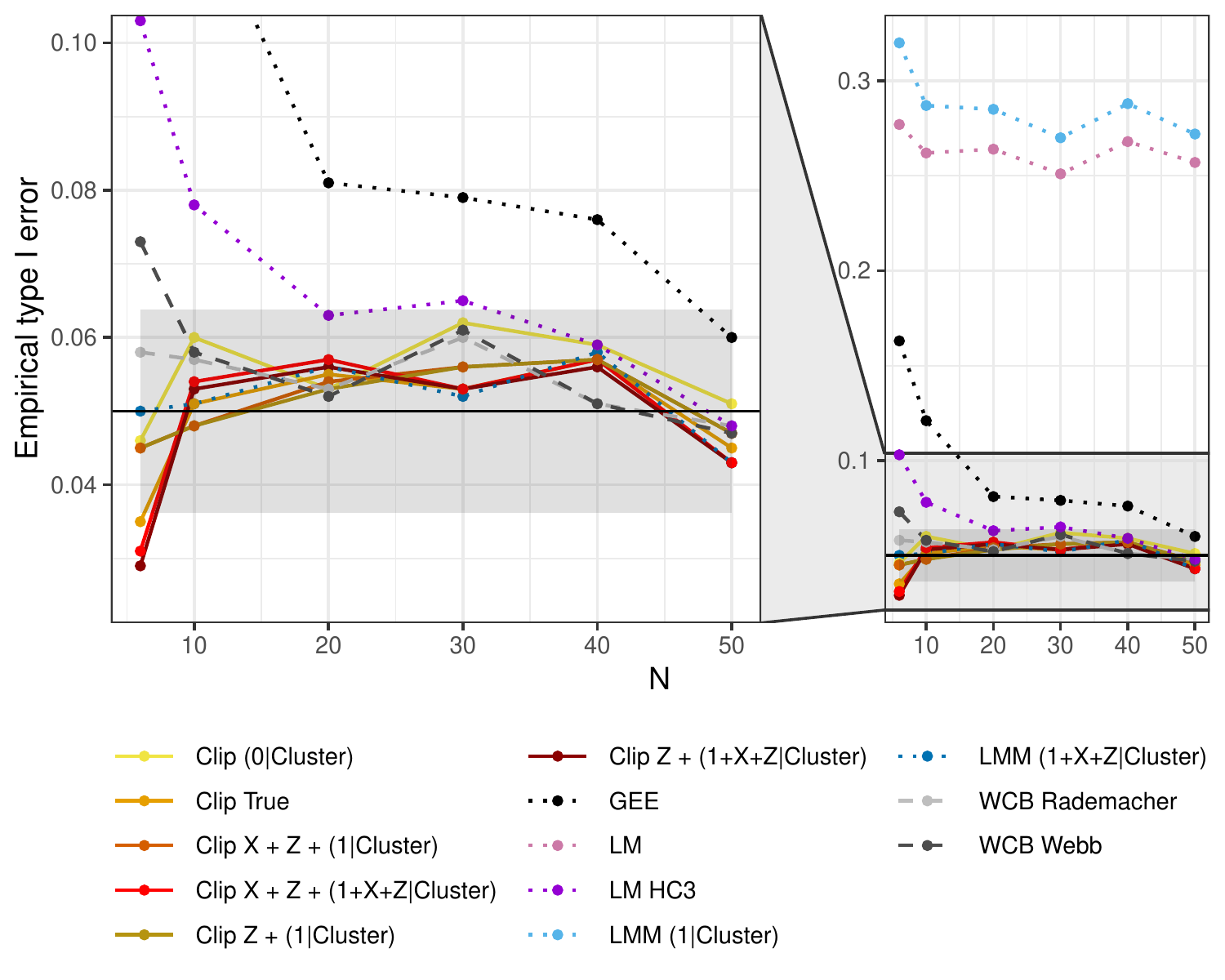} \caption{Estimated type I error considering $N \in \{6, 10, 20,30,40,50\}$ number of clusters and $n_j \sim \text{Uniform}(10,20)$ repeated measurements. Each line represents one model, and the grey area around the solid horizontal black line represents the $0.95$ confidence bound for $\alpha = 0.05$.} \label{supp-fig:error_wc} \end{figure}

\begin{figure}[tbp] \centering \includegraphics[width=.70\linewidth]{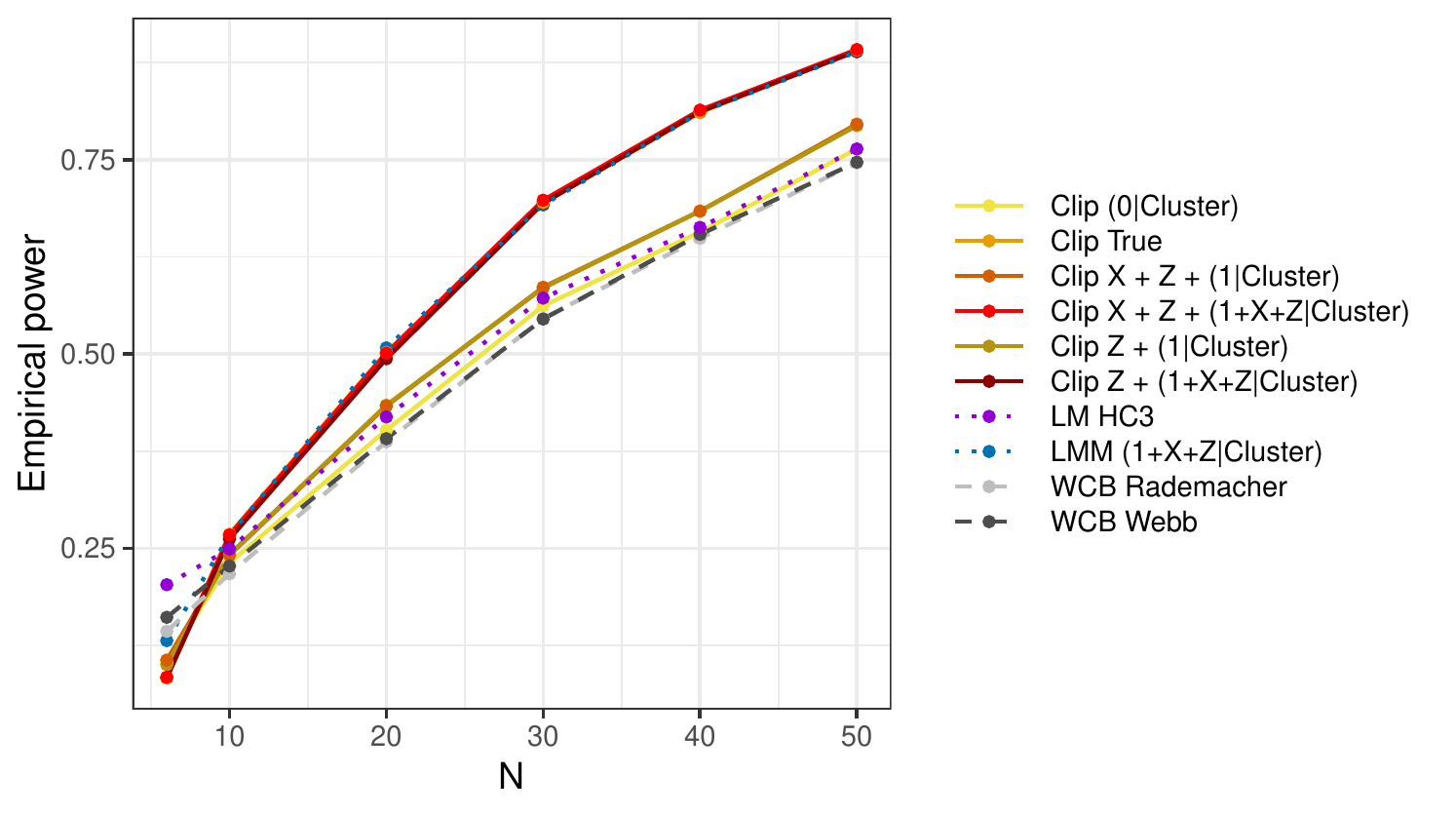} \caption{Estimated power considering $N \in \{6, 10, 20,30,40,50\}$ number of clusters and $n_j \sim \text{Uniform}(10,20)$ repeated measurements. Each line represents one model.} \label{supp-fig:power_wc} \end{figure}

\subsection{Univariate skewed random effects}\label{supp-app:skewed_re}

Here, the Gaussian random-effects distribution is replaced by a zero-centered multivariate skew-normal distribution. Specifically, $\boldsymbol b_j\sim \operatorname{SN}_3(\boldsymbol\xi,\boldsymbol\Omega,\boldsymbol\alpha)$, with $\boldsymbol\alpha=(4,2,5)^\top$, where $\boldsymbol\Omega$ is the dispersion matrix used in Section~\ref{subsec:sim_uni}. The location vector $\boldsymbol\xi$ is chosen so that $\operatorname{E}(\boldsymbol b_j)=\boldsymbol 0$.

Figures~\ref{supp-fig:error_skewed} and~\ref{supp-fig:power_skewed} show patterns similar to those observed under Gaussian random effects in Section~\ref{subsec:sim_uni}.

\begin{figure}[tbp] \centering \includegraphics[width=.75\linewidth]{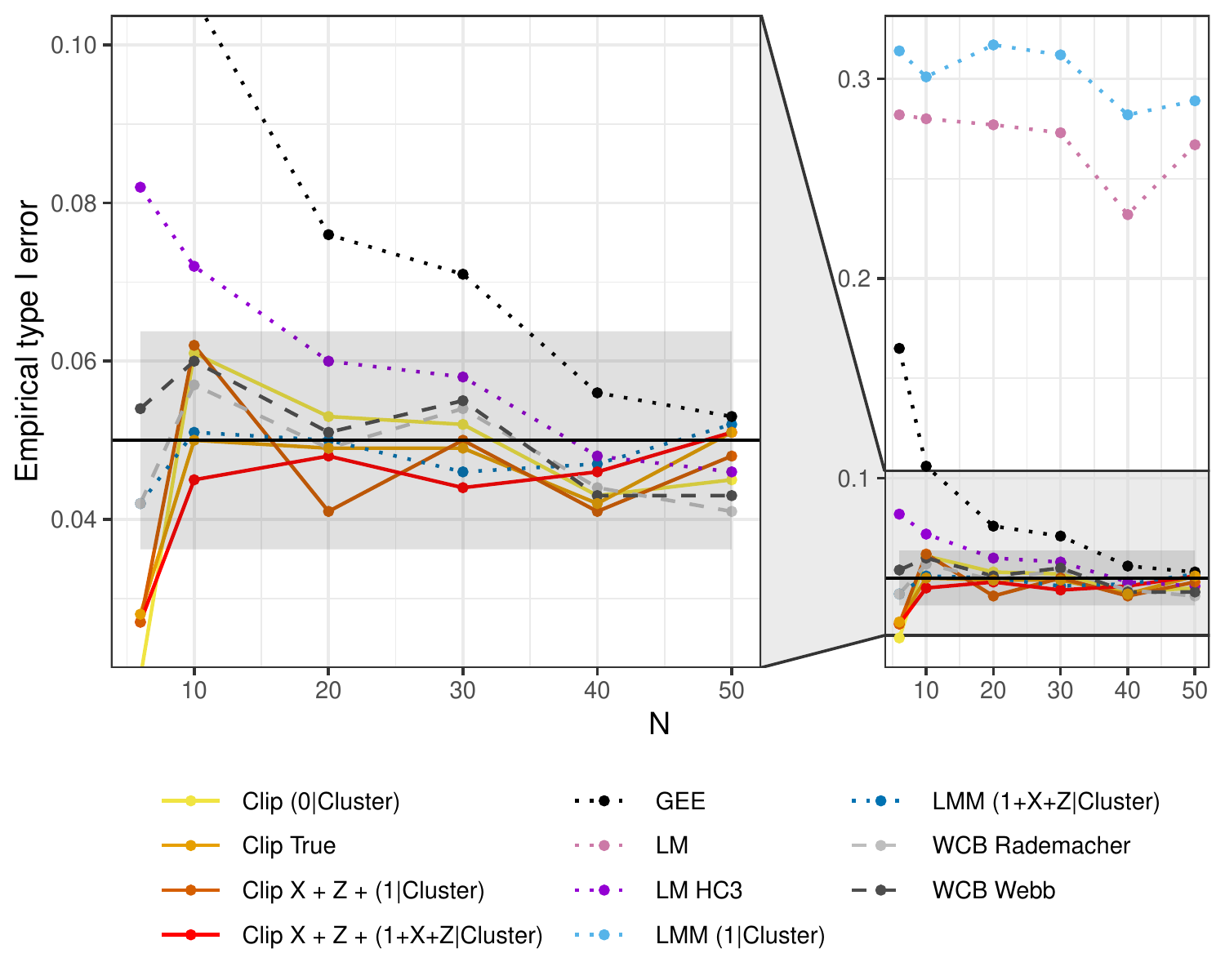} \caption{Estimated type I error under zero-centered skew-normal random effects considering $N \in \{6, 10, 20,30,40,50\}$ number of clusters and $n_j \sim \text{Uniform}(10,20)$ repeated measurements. Each line represents one model, and the grey area around the solid horizontal black line represents the $0.95$ confidence bound for $\alpha = 0.05$.} \label{supp-fig:error_skewed} \end{figure}

\begin{figure}[tbp] \centering \includegraphics[width=.70\linewidth]{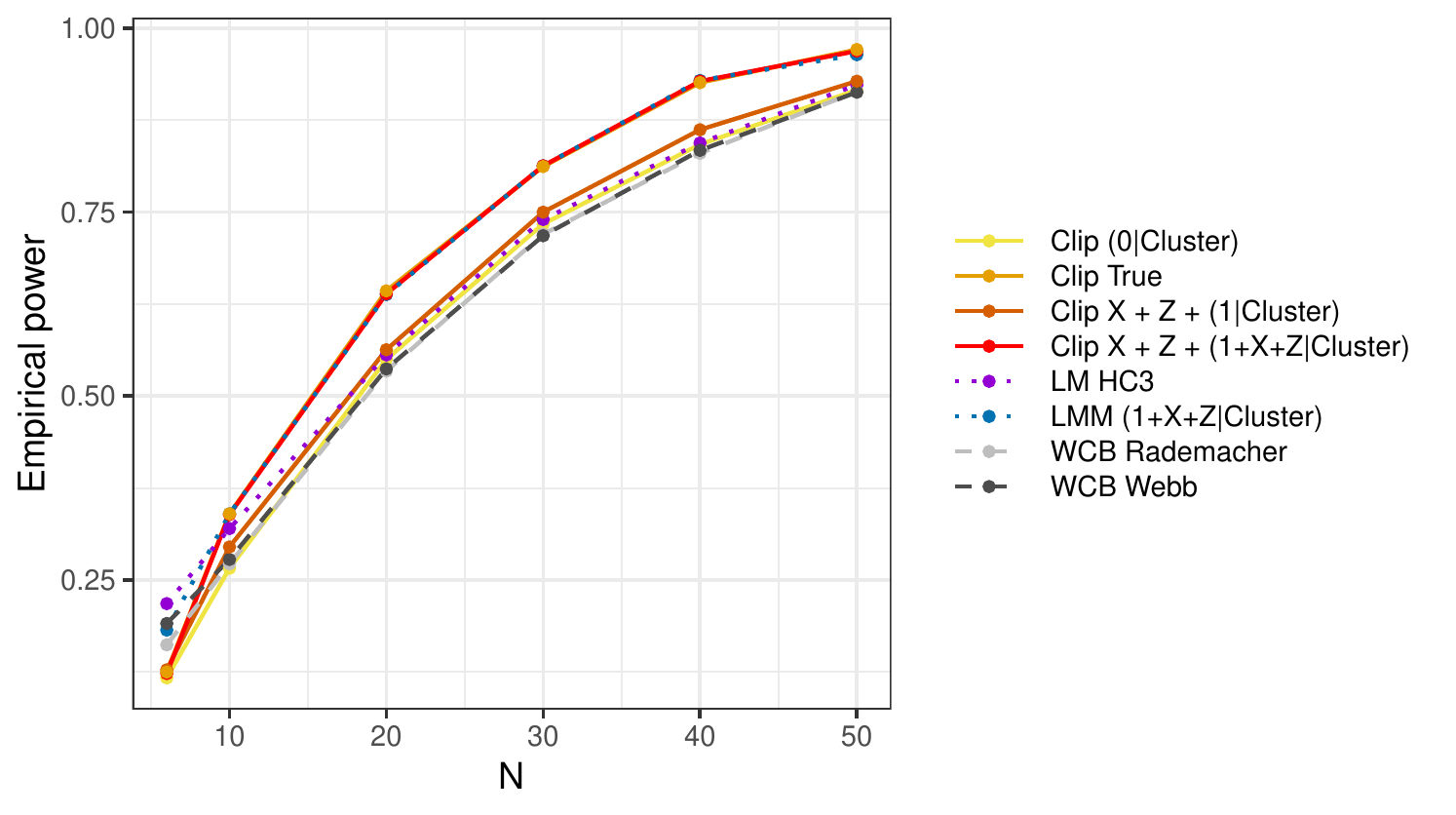} \caption{Estimated power under zero-centered skew-normal random effects considering $N \in \{6, 10, 20,30,40,50\}$ number of clusters and $n_j \sim \text{Uniform}(10,20)$ repeated measurements. Each line represents one model.} \label{supp-fig:power_skewed} \end{figure}

\subsection{Univariate small $N$ large $n_j$}\label{supp-app:small_N_large_n}

Here $N=10$ and, within each scenario, all clusters have common size $n_j\in\{100,150,200\}$.

Figure~\ref{supp-fig:error_sNln} shows that the \emph{clip} procedures, full LMM, and both WCB variants remain close to the nominal level throughout the considered range of $n_j$. The LM HC3 procedure is moderately anti-conservative, whereas GEE, the conventional LM, and the random intercept LMM procedures displayed in the enlarged panel show more pronounced Type I error inflation.

As shown in Figure~\ref{supp-fig:power_sNln}, increasing the common cluster size from $100$ to $200$ produces only limited power gains. Under strong within-cluster dependence, the effective amount of independent information is largely determined by the number of clusters rather than by the total number of observations.

\begin{figure}[tbp] \centering \includegraphics[width=.75\linewidth]{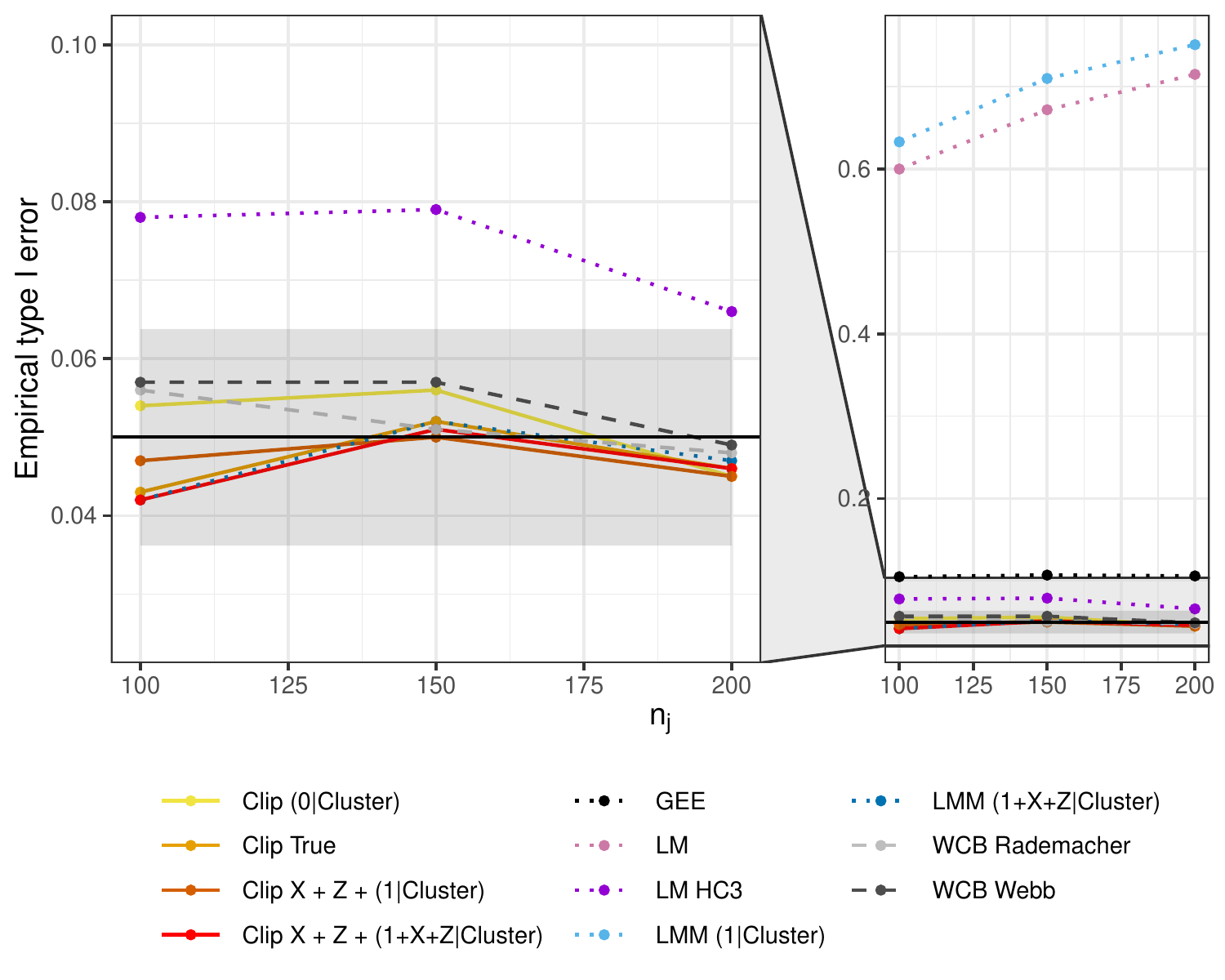} \caption{Estimated Type I error with $N=10$ clusters and $n_j\in\{100,150,200\}$. Each line represents one model, and the grey area around the solid horizontal black line represents the $0.95$ confidence bound for $\alpha = 0.05$.} \label{supp-fig:error_sNln} \end{figure}

\begin{figure}[tbp] \centering \includegraphics[width=.70\linewidth]{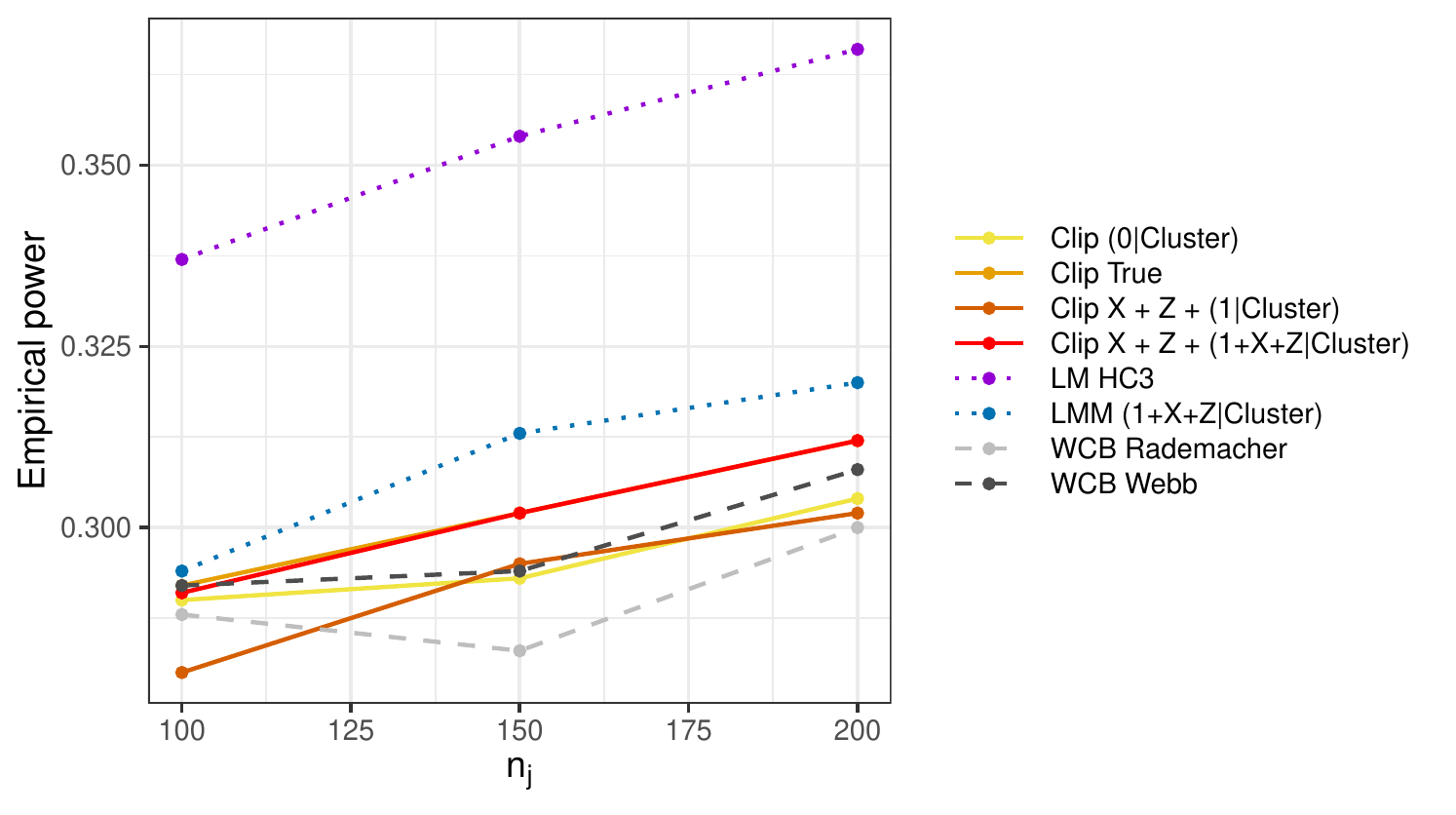} \caption{Estimated power with $N=10$ clusters and $n_j\in\{100,150,200\}$. Each line represents one model. } \label{supp-fig:power_sNln} \end{figure}

\subsection{Univariate highly unbalanced cluster sizes}\label{supp-app:unbalanced_clusters}

Here, $N\in\{10,20,30,40,50\}$ and for each value of $N$, $20\%$ of the clusters contain $150$ observations, whereas the remaining $80\%$ contain only $5$ observations.

Figure~\ref{supp-fig:error_sb} shows that the \emph{clip} procedures remain approximately calibrated despite the pronounced imbalance in cluster sizes. The WCB procedures are more sensitive when the number of clusters is small: Rademacher weights produce some over-rejection for $N=10$, whereas Webb weights are conservative in the same configuration. For some larger values of $N$, WCB Webb instead shows mild over-rejection. GEE, the conventional LM, the LM HC3 procedure, and the random-intercept LMM exhibit substantial Type I error inflation, particularly for small $N$. The LMM based on the full random-effects specification performs substantially better, although some finite-sample distortion remains in the smallest configurations.

Figure~\ref{supp-fig:power_sb} highlights the efficiency gain obtained by using an informative working covariance. The \emph{clip} procedure based on the full random-effects working model achieves power close to the oracle procedure and to the full LMM. In contrast, working independence, the random-intercept working covariance, and the two WCB procedures are considerably less powerful. Therefore, although the validity of \emph{clip} does not require correct specification of the working covariance, a suitable working model can provide substantial efficiency gains under severe cluster-size imbalance.

\begin{figure}[tbp] \centering \includegraphics[width=.75\linewidth]{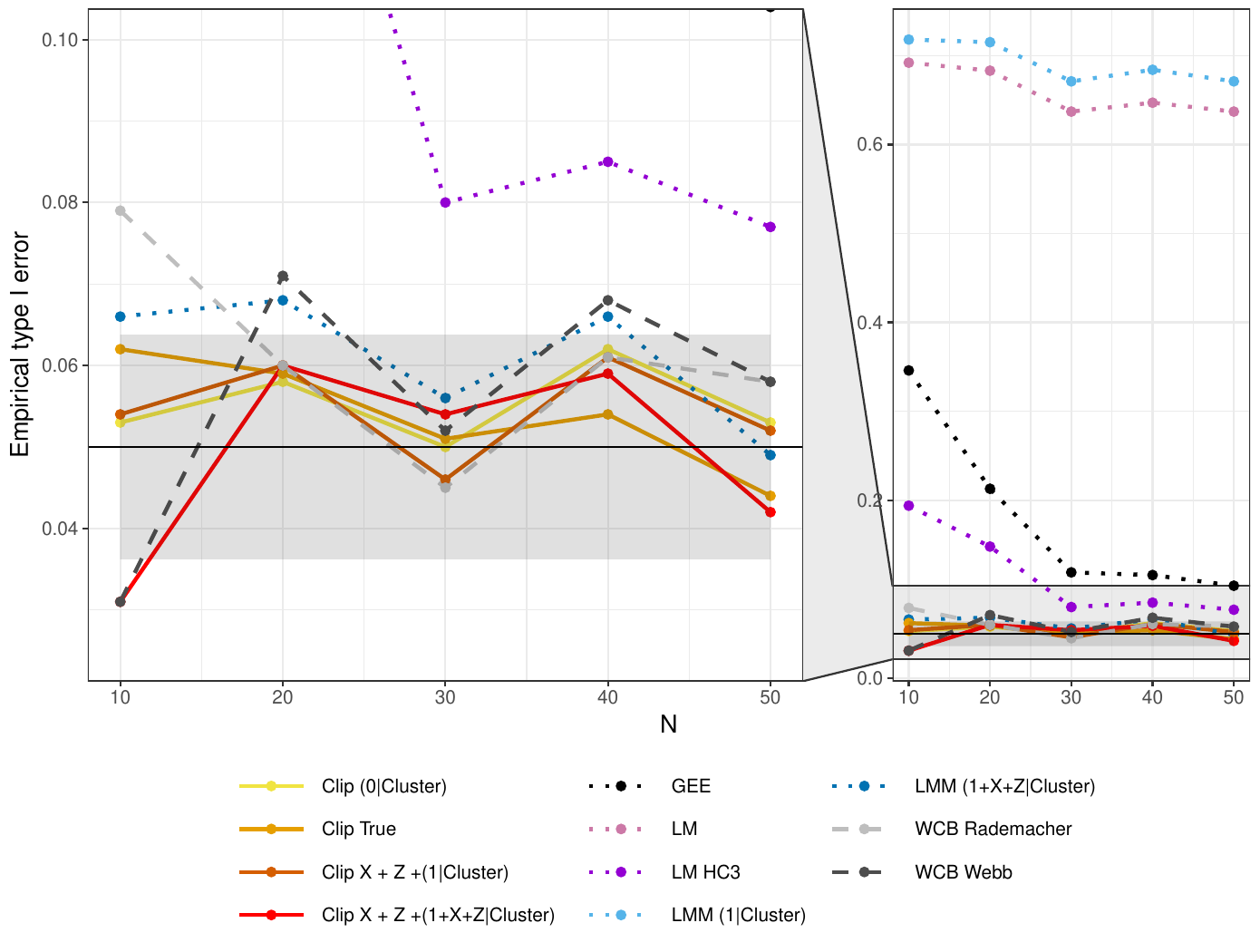} \caption{Estimated Type I error for $N\in\{10,20,30,40,50\}$ highly unbalanced clusters. Each line represents one model, and the grey area around the solid horizontal black line represents the $0.95$ confidence bound for $\alpha = 0.05$.} \label{supp-fig:error_sb} \end{figure}

\begin{figure}[tbp] \centering \includegraphics[width=.70\linewidth]{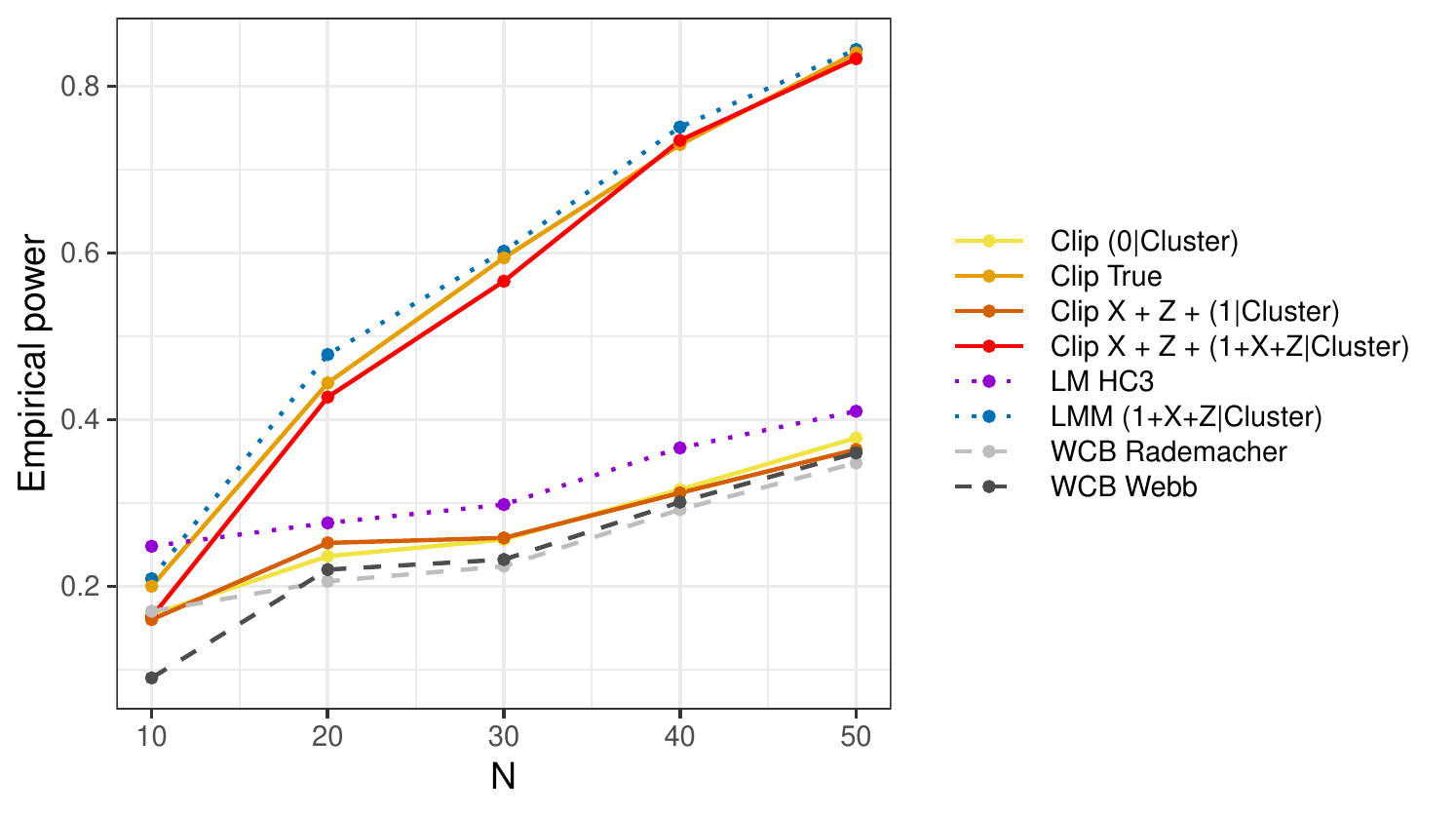} \caption{Estimated power for $N\in\{10, 20,30,40,50\}$ highly unbalanced clusters. Each line represents one model.} \label{supp-fig:power_sb} \end{figure}

Overall, these additional experiments show that increasing the number of observations within a fixed and small number of clusters cannot generally replace increasing the number of independent clusters. They also demonstrate that pronounced cluster-size heterogeneity can substantially affect both analytical methods and WCB inference. Across the considered settings, \emph{clip} provides comparatively stable Type I error control, whereas its power depends strongly on the informativeness of the chosen working covariance.

\section{Real data analysis}\label{supp-appC}

Figure~\ref{supp-fig:eeg2} displays the adjusted $p$-values for the \emph{scramble} factor, in the same format as Figure~\ref{fig:eeg}. The interaction effect is not shown because no significant evidence was found across electrodes after multiplicity adjustment.

\begin{figure}[h] \centering \includegraphics[width=\textwidth]{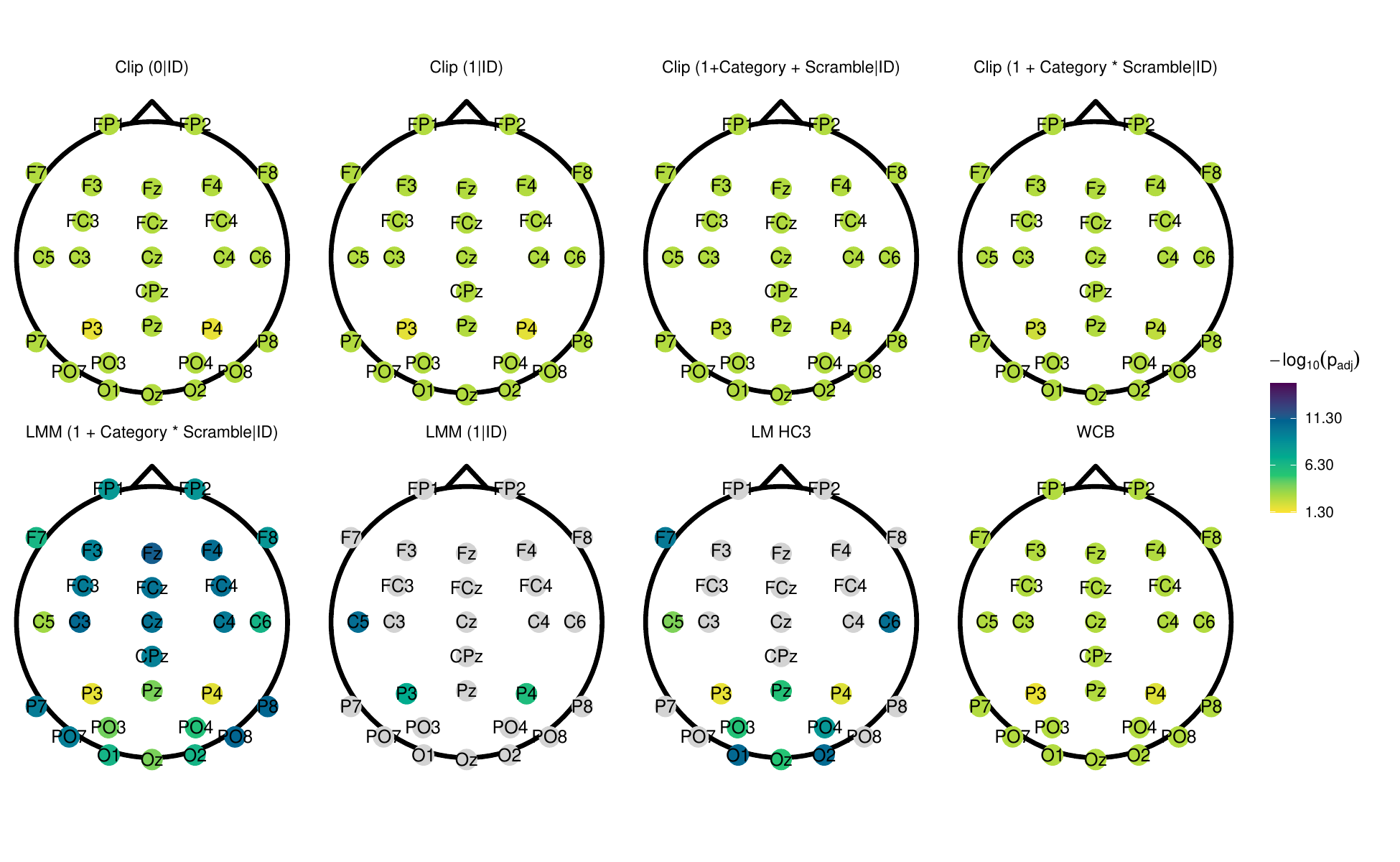} \caption{Adjusted $p$-values for the \emph{scramble} effect for each electrode, estimated using the \emph{clip} approach (top row) and the LMM/LM HC3/WCB methods (bottom row). The corresponding working covariance matrix and random-effects structure are indicated in parentheses.} \label{supp-fig:eeg2} \end{figure} 
\clearpage

% One bibliography for both parts.
\bibliographystyle{abbrvnat}
\bibliography{bibliography}

\end{document}